\documentclass[12pt]{article}

\usepackage{amssymb}
\usepackage{amsmath}
\usepackage{amsthm}
\usepackage{mathrsfs}
\usepackage{MnSymbol}
\usepackage{cite}
\usepackage{graphicx}
\usepackage{appendix}
\usepackage{enumerate}
\usepackage{url}

\newcommand{\rnet}{\{$\mathscr{S,C,R}\}\,$}
\newcommand{\rnetbar}{\{$\mathscr{S,\bar{C},\bar{R}}\}\,$}
\newcommand{\RS}{\ensuremath{\mathbb{R}^{\mathscr{S}}}}
\newcommand{\RSst}{\ensuremath{\mathbb{R}^{\mathscr{S^*}}}}
\newcommand{\RU}{\ensuremath{\mathbb{R}^{\mathscr{U}}}}

\newcommand{\PS}{\ensuremath{\mathbb{R}_+^{\mathscr{S}}}}
\newcommand{\PR}{\ensuremath{\mathbb{R}_+^{\mathscr{R}}}}

\newcommand{\RR}{\ensuremath{\mathbb{R}^{\mathscr{R}}}}
\newcommand{\PbarS}{\ensuremath{\overline{\mathbb{R}}_+^{\mathscr{S}}}}
\newcommand{\PbarU}{\ensuremath{\overline{\mathbb{R}}_+^{\mathscr{U}}}}
\newcommand{\PbarUst}{\ensuremath{\overline{\mathbb{R}}_{+*}^{\mathscr{U}}}}
\newcommand{\kinsys}{\{$\mathscr{S,C,R,K}\}$}

\newcommand{\scrS}{\ensuremath{\mathscr{S}}}
\newcommand{\scrC}{\ensuremath{\mathscr{C}}}
\newcommand{\scrR}{\ensuremath{\mathscr{R}}}
\newcommand{\scrU}{\ensuremath{\mathscr{U}}}
\newcommand{\scrRst}{\ensuremath{\mathscr{R}^*}}
\newcommand{\scrSst}{\ensuremath{\mathscr{S}^*}}
\newcommand{\scrRo}{\ensuremath{\mathscr{R}_0}}
\newcommand{\scrSo}{\ensuremath{\mathscr{S}_0}}
\newcommand{\scrK}{\ensuremath{\mathscr{K}}}

\newcommand{\rxn}{\ensuremath{y \to y'}}
\newcommand{\supp}{\ensuremath{\mathrm{supp} \,}}
\newcommand{\sgn}{\ensuremath{\mathrm{sgn} \,}}

\theoremstyle{plain}
\newtheorem{theorem}{Theorem}[section]

\newtheorem{lemma}[theorem]{Lemma}
\newtheorem{proposition}[theorem]{Proposition}
\newtheorem{corollary}[theorem]{Corollary}
\theoremstyle{definition}
\newtheorem{definition}[theorem]{Definition}
\newtheorem{example}[theorem]{Example}

\theoremstyle{remark}
\newtheorem{rem}[theorem]{Remark}

\begin{document}

\author{Guy Shinar\thanks{InBrain Therapeutics Ltd., 12 Metzada St., Ramat Gan 52235, Israel.  E-mail: shinarg@gmail.com. Work initiated while GS was in the Department of Molecular Cell Biology, Weizmann Institute of Science, Rehovot 76100, Israel.}
\and Martin Feinberg\thanks{The William G. Lowrie Department of Chemical \& Biomolecular Engineering and Department of Mathematics, Ohio State University, 140 W. 19th Avenue, Columbus, OH, USA 43210.  E-mail: feinberg.14@osu.edu. MF was supported by NSF grant EF-1038394 and NIH grant 1R01GM086881-01.} \thanks{Corresponding author. Phone: 614-688-4883.}} 

\title{\emph{Concordant Chemical Reaction Networks and the Species-Reaction Graph}}

\date{\emph{\today}}
\maketitle
\begin{abstract}In a recent paper it was shown that, for chemical reaction networks possessing a subtle structural property called \emph{concordance}, dynamical behavior of a very circumscribed (and largely stable) kind is enforced, so long as the kinetics lies within the very broad and natural \emph{weakly monotonic} class. In particular, multiple equilibria are precluded, as are degenerate positive equilibria. Moreover, under certain circumstances, also related to concordance, all real eigenvalues associated with a positive equilibrium are negative. Although concordance of a reaction network can be decided by readily available computational means, we show here that, when a nondegenerate network's \emph{Species-Reaction Graph} satisfies certain mild conditions, concordance and its dynamical consequences are ensured. These conditions are weaker than earlier ones invoked to establish kinetic system injectivity, which, in turn, is just one ramification of network concordance.   Because the Species-Reaction Graph resembles pathway depictions often drawn by biochemists, results here expand the possibility of inferring significant dynamical information directly from standard biochemical reaction diagrams.
\end{abstract}
\newpage
\tableofcontents
\newpage
\section{Introduction}
\label{sec:Intro}
\subsection{Background}
\label{subsec:Background}
	This article is intended as a supplement to another one \cite{shinar_concordant_2011}, in which we defined the large class of \emph{concordant chemical reaction networks}\footnote{The formal definitions of concordance and strong concordance appear here in Sections \ref{sec:ConcordantNetworks} and \ref{sec:StrongConcordance}.} and deduced many of the striking properties common to all members of that class. We argued that, so long as the kinetics is weakly monotonic ($\S$\ref{subsec:Definitions}), network concordance enforces behavior of a very circumscribed kind. 

	Among other things, we showed that the class of concordant networks coincides \emph{precisely} with the class of networks which, when taken with any weakly monotonic kinetics, invariably give rise to kinetic systems that are \emph{injective} --- a quality that precludes, for example, the possibility of switch-like transitions between distinct positive stoichiometrically compatible steady states.  Moreover, we showed that certain properties of concordant networks taken with weakly monotonic kinetics extend to still broader categories of kinetics --- including kinetics that involve product inhibition  --- provided that the networks considered are not only concordant but also   \emph{strongly concordant} \cite{shinar_concordant_2011}. 

	Although reaction network concordance is a subtle structural property, determination of whether or not a network is concordant (or strongly concordant) is readily accomplished with the help of easy-to-use and freely available software \cite{Ji_toolboxWinV21}, at least if the network is of moderate size.\footnote{Some of the underlying algorithms are given in \cite{Ji_Thesis}.} In this way, one can easily determine whether the several dynamical consequences of concordance or strong concordance accrue to a particular reaction network of interest.

\subsection{The role of the Species-Reaction Graph}

	The \emph{Species-Reaction Graph} (SR Graph), defined in Section \ref{sec:MainTheorems}, is a graphical depiction of a chemical reaction network resembling pathway diagrams often drawn by biochemists. For so-called ``fully open" systems ($\S$\ref{subsec:SRFullyOpen}), earlier work \cite{craciun_multiple_2006-1,craciun_understanding_2006-1,banaji_graph-theoretic_2009,banaji_graph-theoretic_2010}  indicated that, when the SR Graph satisfies certain structural conditions and when the kinetics is within a specified class, the governing differential equations can only admit behavior of a restricted kind. A survey of some of that earlier work, beginning with the Ph.D. research of Paul Schlosser\cite{Schlosser_PhDThesis,schlosser_theory_1994} is provided in \cite{shinar_concordant_2011}. Although most of the initial SR Graph results for fully open networks were focused on mass action kinetics, that changed with the surprising work of Banaji and Craciun \cite{banaji_graph-theoretic_2009,banaji_graph-theoretic_2010}. (In contrast to the SR Graph results, however, the network analysis tools in \cite{Ji_toolboxWinV21} are indifferent to whether the network is fully open.)

	With this as background, in \cite{shinar_concordant_2011} we asserted without proof that attributes of the SR Graph shown earlier to be sufficient for other network properties \cite{banaji_graph-theoretic_2009,craciun_multiple_2006-1,craciun_understanding_2006-1,banaji_graph-theoretic_2010} are, in fact, sufficient to ensure not only concordance but also strong concordance, at least in the ``fully open" setting. In turn, those other properties (e.g., the absence of multiple steady states when the kinetics is weakly monotonic)  largely \emph{derive} from concordance. \emph{That certain SR Graph attributes might ensure concordance seems, then, to be the fundamental idea, with attendant dynamical consequences of those same SR Graph attributes ultimately \emph{descending} from concordance.}

	Although handy computational tools in \cite{Ji_toolboxWinV21} will generally be more incisive than the SR Graph in determining a network's concordance properties and although those computational tools are indifferent to whether the network is fully open, the SR Graph nevertheless has its strong attractions. Not least of these is the close relationship that the SR Graph bears to reaction network diagrams often drawn by biochemists. For this reason alone,  far-from-obvious dynamical consequences that might be inferred from these ubiquitous diagrams take on considerable interest. Moreover, theorems that tie network concordance to properties of the SR Graph point to consequences of reaction network structure that are not easily gleaned from implementation of computational tests on a case-by-case basis.

	Thus, our primary purpose here is to prove the assertions made in \cite{shinar_concordant_2011} about connections between concordance of a network and the nature of its SR Graph. In fact, we go further in two directions, directions that would have been somewhat askew to the main thrusts of \cite{shinar_concordant_2011}:
 
\subsubsection{The Species-Reaction Graph and the concordance of ``fully open" networks}
\label{subsec:SRFullyOpen}
	First, we show that one can infer concordance of a so-called ``fully open" network when its SR graph satisfies conditions, stated in Theorem \ref{thm:ConcordanceTheorem}, that are substantially weaker than those stated in \cite{shinar_concordant_2011} (or in \cite{banaji_graph-theoretic_2009,banaji_graph-theoretic_2010,craciun_multiple_2006-1, craciun_understanding_2006-1});  Corollary \ref{cor:ConcordanceThmCorollary} invokes the SR Graph conditions stated in\cite{shinar_concordant_2011}. A fully open network is one containing a ``degradation reaction" $s\to 0$ for each species $s$ in the network.\footnote{Nevertheless, the SR graph is drawn for the so-called ``true" chemistry, devoid of degradation-synthesis reactions such as $s \to 0$ or $0 \to s$.}  Proof of concordance when the aforementioned SR Graph conditions are satisfied is by far the main undertaking of this article. 

\subsubsection{The Species-Reaction Graph and the concordance of networks that are not ``fully open"}
\label{subsec:NotFullyOpen}

	Second, we examine more completely concordance information that the SR Graph gives for reaction networks that are \emph{not} fully open. 

	For a network that is not fully open, its \emph{fully open extension} is the network obtained from the original one by adding degradation reactions for all species for which such degradation reactions are not already present. In \cite{shinar_concordant_2011} we proved that a \emph{normal} network is concordant (respectively, strongly concordant) if its fully open extension is concordant (strongly concordant). A normal network is one satisfying a very mild condition described in \cite{craciun_multiple_2010},  \cite{shinar_concordant_2011}, and, here again, in $\S$\ref{subsec:normality-weaknormality}. \emph{The large class of normal networks includes all weakly reversible networks and, in particular, all reversible networks \cite{craciun_multiple_2010}.}  

	Thus, if inspection of the SR Graph can serve to establish, by means of  theorems crafted  \emph{for fully open networks}, that a \emph{normal} network's fully open extension is concordant or strongly concordant, \emph{then the original network itself has that same property}. In this way, the seemingly restricted power of ``fully open" SR graph theorems extends far beyond the fully open setting. 

	This has some importance, especially in relation to reversible networks. Consider a reaction network whose SR Graph ensures, by virtue of Corollary \ref{cor:ConcordanceThmCorollary}, concordance of the network's fully open extension.  Then the original network can \emph{fail} to be concordant only if it is not reversible. But a kinetic system based on a network that is not reversible is usually deemed to be an approximation of a more exact ``nearby" kinetic system in which all reactions are reversible, perhaps with some reverse reactions having extremely small rates.  

	To the extent that this is the case, lack of concordance in the original network is an artifact of the approximation:  The reversible network underlying the more exact kinetic system  is concordant, so that system inherits all of the dynamical properties in \cite{shinar_concordant_2011} that are consequences of network concordance. (In particular, it inherits those properties listed as items (i) -- (iii) in Theorem \ref{thm:DynamicalTheorem}.) Thus, any absence of these attributes in a kinetic system based on the original network is, again, an artifact of the approximation.

	Even so, kinetic system models based on networks that are not reversible (or, more generally, weakly reversible) have an intrinsic interest. They are ubiquitous, and one would like to understand their inherent properties. Results in \cite{shinar_concordant_2011} about network normality and concordance already go a long way in this direction, for the class of normal networks extends well beyond the weakly reversible class.

	Here we argue that these same results extend to the still larger class of \emph{nondegenerate} reaction networks, a class that includes all normal networks and, in particular, all weakly reversible networks. A network is nondegenerate if, taken with \emph{some} differentiably monotonic kinetics \cite{shinar_concordant_2011},  the resulting kinetic system admits \emph{even one} positive composition at which the derivative of the species-formation-function is nonsingular\footnote{A positive composition is one at which all species concentrations are (strictly) positive. When we say that the derivative (Jacobian matrix) of the species-formation rate function is nonsingular we mean that its null space contains no nonzero vector of the stoichiometric subspace.}; otherwise we say that the network is \emph{degenerate} (Definition \ref{def:nondegenerate}). 

	In fact,  degenerate reaction networks are \emph{never} concordant (\S\ref{subsec:normality-weaknormality}), so for them questions about the possibility of concordance are moot.   On the other hand,  a\emph{ nondegenerate} network is concordant (strongly concordant) if the network's fully open extension is concordant (strongly concordant) ($\S$\ref{subsec:ThmsOnSmallToLarge}). \emph{Thus, SR Graph theorems that give information about the concordance or strong concordance of a nondegenerate network's fully open extension give that same information about the network itself.}

	Implications of network normality were considered in some depth in \cite{shinar_concordant_2011}. Because corresponding considerations of the broader notion of nondegeneracy are similar and because those considerations are rather different in spirit from the largely graph-theoretical aspects of most of this article, we chose to defer the entire discussion of nondegeneracy to Section \ref{sec:Extensions}, which includes computational tests whereby normality and nondegeneracy of a network can be affirmed.

	In any case, it should be kept in mind that the computational tools provided in \cite{Ji_toolboxWinV21} are indifferent to whether the network under study is fully open (or nondegenerate).

\subsection{The Species-Reaction Graph and  consequences of concordance}
\label{subsec:ConsequencesOfConcordance}
 For the most part the theorems in this article have as their objective the drawing of connections between the concordance of a reaction network and properties of the network's Species-Reaction Graph. However, it is important to keep in mind the  dynamical consequences of these theorems. When an SR Graph theorem asserts concordance of a particular network, \emph{the same theorem is also asserting that the network  inherits all of the properties shown in \cite{shinar_concordant_2011} to accrue to all concordant networks.} 

	Thus, for example, a theorem that asserts concordance for any reaction network whose SR Graph has Property X is \emph{also} asserting that, for any reaction network whose SR Graph has Property X, \emph{there is no possibility of two distinct stoichiometrically compatible equilibria, at least one of which is positive, no matter what the kinetics might be, so long as it is weakly monotonic}.  

	A theorem that does make explicit connections between the SR Graph and dynamical consequences is offered at the close of the next section.

\section{Main theorems}
\label{sec:MainTheorems}
	The \emph{Species-Reaction Graph (SR Graph)} of a chemical reaction network is a bipartite graph constructed in the following way: The vertices are of two kinds --- \emph{species vertices} and \emph{reaction vertices}. The species vertices are simply the species of the network. The reaction vertices are the reactions of the network but with the understanding that a reversible reaction pair such as $A + B \rightleftarrows P$ is identified with a single vertex. If a  species appears in a particular reaction, then an edge is drawn that connects the species with that reaction, and the edge is labeled with the name of the complex in which the species appears. (The \emph{complexes} of a reaction network are the objects that appear before and after the reaction arrows. Thus, the complexes of reaction $A+B \to P$ are $A+B$ and $P$.) The \emph{stoichiometric coefficient of an edge} is the stoichiometric coefficient of the adjacent species in the labeling complex. We show the Species-Reaction Graph for network \eqref{eq:NiceNetworkExample} in Figure \ref{fig:FirstSRGraphExample}. The arrows on some of the edges will be explained shortly.

\begin{eqnarray}
\label{eq:NiceNetworkExample}
A + B &\rightleftarrows & P \nonumber\\
B + C &\rightleftarrows & Q \\
2A     &\to  & C \nonumber\\
C + D & \to &Q +E \nonumber 
\end{eqnarray}

\begin{figure}[h]
\centering
\includegraphics[scale=.43]{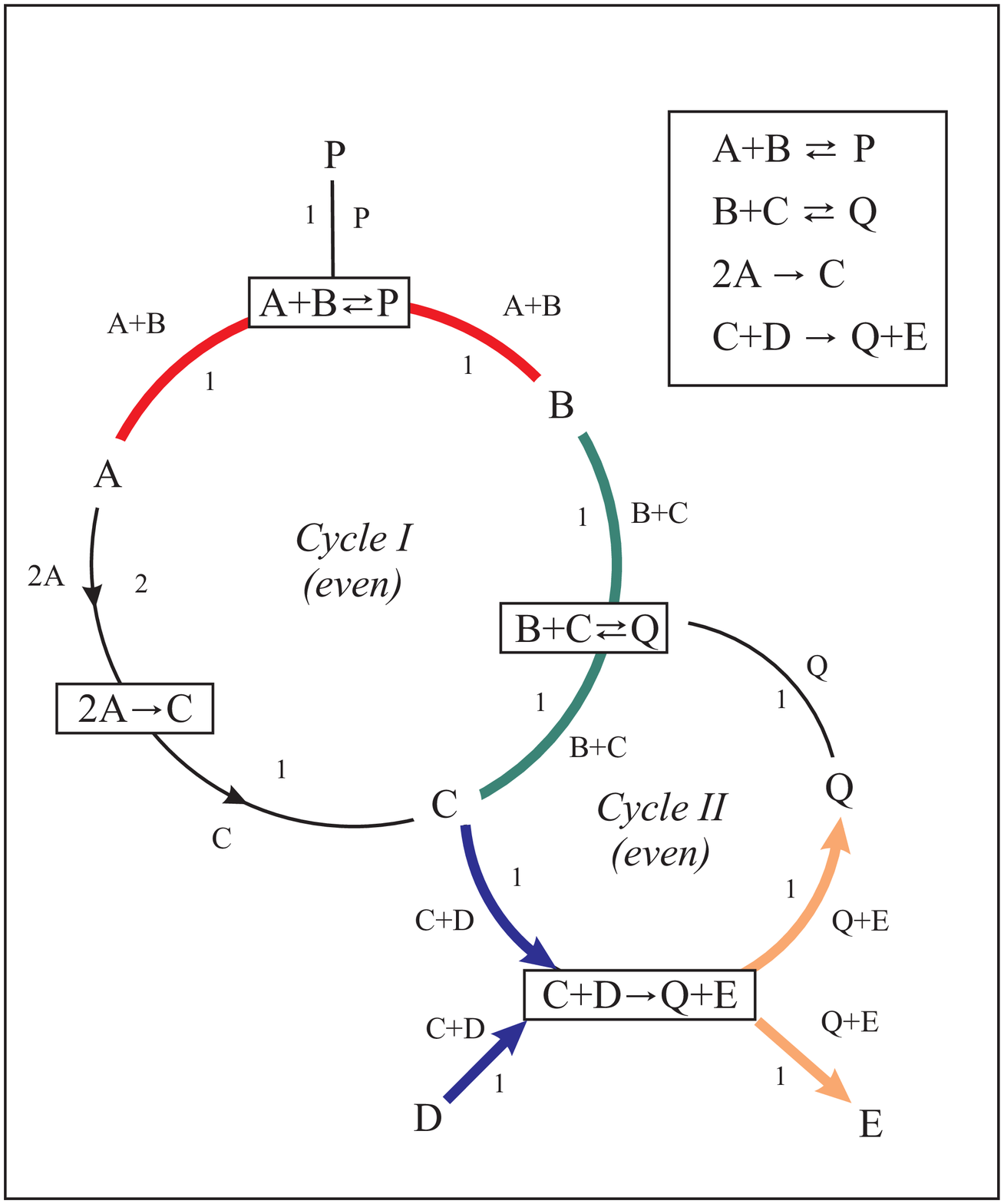}
\caption{An example of a Species-Reaction Graph}
\label{fig:FirstSRGraphExample}
\end{figure}

		It will be understood that our focus is on  concordance of a ``fully open" reaction network \rnet. To say that the network is fully open is to say that, for each species $s$ in the network  there is a ``degradation reaction" of the form $s \to 0$. (There might also be ``synthesis reactions" of the form $0 \to s$.) Moreover, we shall also suppose hereafter that a species can appear on only one side of a reaction. The first of these requirements was invoked in \cite{banaji_graph-theoretic_2010,craciun_multiple_2005-1,craciun_multiple_2006-1}, while the second was invoked only in \cite{banaji_graph-theoretic_2010}.

	Even in consideration of fully open reaction networks we will always restrict our attention to the Species-Reaction Graph for the ``true chemistry" --- that is, the Species-Reaction Graph for the original fully open network but with degradation-synthesis ``reactions" such as $A \to 0$ or $0 \to A$ omitted. As a reminder we will sometimes refer to the ``true-chemistry Species-Reaction Graph," but this understanding will always be implicit even when it is not made explicit.

	By the \emph{intersection of two cycles} in the Species-Reaction Graph, we mean the subgraph consisting of all vertices and edges common to the two cycles.
We say that two cycles have a \emph{species-to-reaction intersection (S-to-R intersection)} if their intersection is not empty and each of the connected components of the intersection is a path having a species at one end and a reaction at the other. In the phrase ``S-to-R intersection" no directionality is implied. In Figure \ref{fig:FirstSRGraphExample}, Cycle I and Cycle II have an S-to-R intersection consisting of the edge connecting species C to reactions \mbox{$B+C \rightleftarrows Q$}. There is a third unlabeled ``outer" cycle, hereafter called Cycle III, that traverses species $A, C, Q, B$ and returns to $A$. It has an S-to-R intersection with Cycle I, in particular, the long outer path connecting reactions \mbox{$B+C  \rightleftarrows Q$} to species $C$ via species $A$. Cycle III also has an S-to-R intersection with Cycle II.

	An \emph{edge-pair} in the SR Graph is a pair of edges adjacent to a common \emph{reaction} vertex. A \emph{c-pair} (abbreviation for \emph{complex-pair}) is an edge-pair whose edges carry the same complex label. For readers with access to color, the c-pairs in Figure \ref{fig:FirstSRGraphExample} are given distinct colors.

	An \emph{even cycle} in the SR Graph is a cycle whose edges contain an even number of c-pairs.\footnote{In the graph theory literature, the term ``even cycle" sometimes refers to a cycle containing an even number of edges. That is never the usage here. In a bipartite graph --- in particular in an SR graph ---  cycles \emph{always} have an even number of edges.} Cycle I contains two c-pairs and Cycle II contains none, so both are even. Cycle III, the large outer cycle, contains one c-pair, so it is not even.

	A \emph{fixed-direction edge-pair} is an edge-pair that is not a c-pair and for which the reaction of the edge-pair is irreversible. In this case, we assign a \emph{fixed direction} to each of the two edges in the intuitively obvious way: The edge adjacent to the reactant species (i.e., the species appearing in the reactant complex) is directed away from that species, and the edge adjacent to the product species (i.e., the species appearing in the product complex) points toward that species. Thus, for example, a fixed-direction edge-pair such as

\begin{equation}
A \quad \overset{A+B}{\mbox{-----}} \quad \fbox{$A + B \to C$} \quad \overset{C}{\mbox{-----}} \quad C \nonumber
\end{equation}
has the fixed direction
\begin{equation}
A \quad \overset{A+B}{\longrightarrow} \quad \fbox{$A + B \to C$} \quad \overset{C}{\longrightarrow} \quad C. \nonumber
\end{equation}
In Figure \ref{fig:FirstSRGraphExample} there are fixed-direction edge-pairs centered at the reactions $2A \to C$ and $C+D \to Q+E$; the corresponding fixed-directions for the adjacent edges are shown in the figure.

	An \emph{orientation} for a simple cycle $s_1R_1s_2R_2...s_nR_ns_1$ in the SR Graph is an assignment of one of two directions to the edges, either
\begin{equation}
s_1 \rightarrow R_1 \rightarrow s_2 \rightarrow R_2...\rightarrow s_n\rightarrow R_n \rightarrow s_1 \nonumber
\end{equation}
or
\begin{equation}
s_1 \leftarrow R_1 \leftarrow s_2 \leftarrow R_2...\leftarrow s_n\leftarrow R_n \leftarrow s_1, \nonumber
\end{equation}

\noindent
that is consistent with any fixed directions the cycle might contain. A cycle is \emph{orientable} if it admits at least one orientation.  By an \emph{oriented cycle} in the SR Graph we mean an orientable cycle taken with a choice of orientation.

	Note that a cycle that has no fixed-direction edge-pair will be orientable and have two orientations. (This happens when  every reaction in the cycle is either reversible or else is contained in a c-pair within the cycle.) A cycle that has just one fixed-direction edge-pair is orientable and has a unique orientation. A cycle that has more than one fixed-direction edge-pair
might or might not be orientable, depending on whether the fixed-direction edge-pairs all point ``clockwise" or ``counterclockwise," but if there is an orientation it will be unique. All cycles in Figure \ref{fig:FirstSRGraphExample} are orientable, with each cycle having a unique orientation.

	A set of cycles in an SR graph \emph{admits a consistent orientation} if the cycles can all be assigned orientations such that each edge contained in one of the cycles has the same orientation-direction in every cycle of the set in which that edge appears. To see that a set of cycles might not admit a consistent orientation, consider a pair of cycles sharing a common edge, each cycle having the same compulsory orientation (e.g., ``clockwise"). In fact, this is precisely the situation for Cycles I and II, so the set consisting of those two cycles does not admit a consistent orientation. On the other hand the set consisting of Cycles I and III (the large outer cycle) does admit a consistent orientation, as does the set consisting of Cycles II and III.

	A \emph{critical subgraph} of an SR Graph is a union of a set of \emph{even} cycles taken from the SR Graph that admits a consistent orientation. Because Cycle III is not even, it is easy to see that there are only two critical subgraphs in Figure \ref{fig:FirstSRGraphExample}: One consists only of Cycle I and the other consists only of Cycle II.

	Consider an oriented cycle in the SR Graph, $s_1 \to R_1 \to s_2\to R_2...\to s_n \to R_n \to s_1$. For each directed reaction-to-species edge $R \to s$ we denote by $f_{R \to s}$ the stoichiometric coefficient associated with that edge. For each directed species-to-reaction edge, we denote by  $e_{s \to R}$ the stoichiometric coefficient associated with it. The cycle is \emph{stoichiometrically expansive} relative to the given orientation if
\begin{equation}
\label{eq:StoichExpansiveCondition}
\frac{f_{R_1 \to s_2}f_{R_2 \to s_3} \dots f_{R_n \to s_1}}{e_{s_1 \to R_1}e_{s_2 \to R_2} \dots e_{s_n \to R_n}} > 1.
\end{equation}
\smallskip

	One of the major goals of this article is proof of the following theorem:

\begin{theorem}
\label{thm:ConcordanceTheorem}
 A fully open reaction network is concordant if its  true-chemistry Species-Reaction Graph has the following properties:
\smallskip

 (i) No even cycle admits a stoichiometrically expansive orientation.
\smallskip

(ii) In no critical subgraph do two even cycles have a species-to-reaction intersection.

\end{theorem}

\begin{rem}
If there are no critical subgraphs --- in particular if there are no orientable even cycles --- then the conditions of the theorem are satisfied trivially, and the network is concordant.
\end{rem}

\begin{rem}
	Determination of which subgraphs of the Species-Reaction Graph are critical requires consideration of orientability.  However, it should be understood that condition $(ii)$ of Theorem \ref{thm:ConcordanceTheorem} is imposed on every \emph{undirected} critical subgraph of the \emph{undirected} Species-Reaction Graph.  In particular, account should be taken of all even cycles in the subgraph, whether or not they be orientable. (Although, a critical subgraph is the union of a set of even cycles that admit a consistent orientation, the resulting subgraph might also contain cycles that are not orientable.) Again, no directionality should be associated with the phrase ``species-to-reaction" intersection.
\end{rem}
\bigskip

	With regard to condition (i), note that a cycle that admits two orientations can be non-expansive relative to both orientations only if, relative to one of the orientations, the number on the left side of  \eqref{eq:StoichExpansiveCondition} is \emph{equal to one} (in which case it will also have that same value with respect to the opposite orientation). 

	In the language of \cite{craciun_multiple_2006-2,craciun_understanding_2006-1,banaji_graph-theoretic_2010}, we say that a (not necessarily orientable) cycle in the SR Graph is an \emph{s-cycle} if relative to an arbitrarily imposed ``clockwise" direction, the calculation on the left side of \eqref{eq:StoichExpansiveCondition} yields the value \emph{one}. With this in mind we can state an ``orientation-free" corollary of Theorem \ref{thm:ConcordanceTheorem}:

\begin{corollary}
\label{cor:ConcordanceThmCorollary}
A fully open reaction network is concordant if its true-chemistry Species-Reaction Graph has the following properties:
\smallskip

 (i) Every even cycle is an s-cycle.
\smallskip

(ii) No two even cycles have a species-to-reaction intersection.

\end{corollary}

	Especially for networks in which there are several irreversible reactions, Theorem \ref{thm:ConcordanceTheorem} is likely to be more incisive than its corollary. Consider, for example, the Species-Reaction Graph for network \eqref{eq:NiceNetworkExample}. \emph{Both} conditions of Corollary \ref{cor:ConcordanceThmCorollary} fail. Nevertheless, \emph{both} of the (weaker) conditions of Theorem \ref{thm:ConcordanceTheorem} are satisfied, so a fully open network having \eqref{eq:NiceNetworkExample} as its true chemistry is concordant.

\begin{rem} We shall see very easily in Section \ref{sec:Extensions} that network \eqref{eq:NiceNetworkExample} is normal. Because every normal network with a concordant fully open extension is itself concordant, we can conclude that network \eqref{eq:NiceNetworkExample}, like its fully open extension, is concordant.
\end{rem}
\bigskip

	When a fully open network \emph{does} satisfy the stronger conditions of Corollary \ref{cor:ConcordanceThmCorollary} we can say more, this time about \emph{strong} concordance.

\begin{theorem}
\label{thm:StrongConcordanceThm}
 A fully open reaction network is strongly concordant if its true-chemistry Species-Reaction Graph has the following properties:
\smallskip

 (i) Every even cycle is an s-cycle.
\smallskip

(ii) No two even cycles have a species-to-reaction intersection.

\end{theorem}

	Although the fully open extension of network \eqref{eq:NiceNetworkExample} is concordant, it is not strongly concordant. (This can be quickly determined by the freely available \emph{Chemical Reaction Network Toolbox} \cite{Ji_toolboxWinV21}.) The Species-Reaction Graph for network \eqref{eq:NiceNetworkExample} does not satisfy the conditions of Theorem \ref{thm:StrongConcordanceThm}.

\begin{rem}	Although these theorems nominally speak in terms of concordance or strong concordance of a fully open network, they are even more generous than they seem: Recall that in  Section \ref{sec:Extensions} we indicate why the concordance properties ensured by the theorems actually extend beyond the fully open setting to a far wider class of reaction networks. In particular, they extend to the class of \emph{nondegenerate} networks described earlier in \S\ref{subsec:NotFullyOpen} (including \emph{all} weakly reversible networks). For nondegenerate  networks, then, the fully open requirement of Theorems \ref{thm:ConcordanceTheorem} and \ref{thm:StrongConcordanceThm} becomes moot. 
\bigskip

\end{rem}

	In the spirit of \S\ref{subsec:ConsequencesOfConcordance} we close this section with the statement of a theorem that ties the hypothesis of Theorem \ref{thm:ConcordanceTheorem} not to concordance itself but, instead, to dynamical consequences of concordance proved in \cite{shinar_concordant_2011}. 	The following theorem is essentially a corollary of Theorem \ref{thm:StrongConcordanceThm}, theorems in \cite{shinar_concordant_2011}, and material to be discussed in Section \ref{sec:Extensions}. Language in the theorem is that used in \cite{shinar_concordant_2011}. Much of it is reviewed in Section \ref{sec:CRNTPreliminaries} of this article.
\bigskip

\begin{theorem}
\label{thm:DynamicalTheorem}
For any nondegenerate reaction network whose true chemistry Species-Reaction Graph satisfies conditions (i) and (ii) of Theorem \ref{thm:ConcordanceTheorem} the following statements hold true: 
\medskip

\noindent
(i) For each choice of weakly monotonic kinetics, every positive equilibrium is unique within its stoichiometric compatibility class. That is, no positive equilibrium is stoichiometrically compatible with a different equilibrium, positive or otherwise.
\medskip

\noindent
(ii) If the network is weakly reversible then, for each choice of kinetics (not necessarily weakly monotonic), no nontrivial stoichiometric compatibility class has an equilibrium on its boundary. In fact, at each boundary composition in any non-trivial stoichiometric compatibility class the species-formation-rate vector points into the stoichiometric compatibility class in the sense that there is an absent species produced at strictly positive rate. If, in addition, the network is conservative then, for any choice of a continuous weakly monotonic kinetics, there is precisely one equilibrium in each nontrivial stoichiometric compatibility class, and it is positive.
\medskip

\noindent
(iii) If the kinetics is differentiably monotonic then every positive equilibrium is nondegenerate. Moreover, every real eigenvalue associated with a positive equilibrium is negative.
\end{theorem}

	As in \cite{shinar_concordant_2011}, it is understood that eigenvalues in the theorem statement are those associated with eigenvectors in the stoichiometric subspace. Similarly, when we say that a positive equilibrium in nondegenerate, we mean that, for the equilibrium, zero is not an eigenvalue corresponding to an eigenvector in the stoichiometric subspace. A stoichiometric compatibility class is nontrivial if it contains at least one positive composition.

\begin{rem} 
\label{rem:AllOrNothing}
Consider a (not necessarily nondegenerate) reaction network whose true chemistry SR Graph satisfies conditions (i) and (ii) of Theorem \ref{thm:ConcordanceTheorem}. With respect to the possibility of nondegenerate positive equilibria, Theorem \ref{thm:DynamicalTheorem} describes something of an all or nothing situation: 
If there is \emph{some} differentiably monotonic kinetics that gives rise to \emph{even one} nondegenerate positive equilibrium, then the network itself is nondegenerate, in which case  \emph{every} positive equilibrium arising from \emph{any} differentiably monotonic kinetics is nondegenerate (and unique within its stoichiometric compatibility class).
\end{rem}

\begin{rem} Consider a nondegenerate reaction network that is not necessarily fully open. If its true-chemistry SR Graph satisfies the conditions of Theorem  \ref{thm:StrongConcordanceThm}, then considerations in Section \ref{sec:Extensions} will indicate that the network is strongly concordant. In this case, the dynamical properties of strongly concordant networks given in \cite{shinar_concordant_2011} accrue to the network at hand, and one can again deduce a theorem which, in the spirit of Theorem \ref{thm:DynamicalTheorem}, makes statements about general properties of kinetic systems the network engenders, this time including those that derive from so-called ``two-way monotonic kinetics."

\end{rem}

\section{Reaction network theory preliminaries}
\label{sec:CRNTPreliminaries}

This section is, for the most part, is a compendium of the definitions and infrastructure used in \cite{shinar_concordant_2011}, repeated here for the reader's convenience. Reference \cite{shinar_concordant_2011} has more in the way of discussion, and  \cite{feinberg_lectureschemical_1979} has still more in the way of motivation. We begin with notation.

\subsection{Notation}

When $I$ is a finite set (for example, a set of species or a set of reactions), we denote the vector space of real-valued functions with domain $I$ by $\mathbb{R}^I$. If $N$ is the number of elements in the set $I$, then $\mathbb{R}^I$ is, in effect, a copy of the standard vector space $\mathbb{R}^N$, with the components of a vector $x \in \mathbb{R}^I$  indexed by the names of the members of $I$ instead of the integers $1,2,\dots,N$. For $x \in \mathbb{R}^I$ and $i \in I$,  the symbol $x_i$ denotes the value (component) of $x$ corresponding to the element $i \in I$. 

	For example, if $\scrS = \{NO, O_2, NO_2\}$ is the set of species in a chemical system and if  $c_{NO}, c_{O_2}$, and $c_{NO_2}$ are the molar concentrations of the three species in a particular mixture state, then that state can be represented by a ``composition vector" $c$ in the vector space \RS. That is, $c$ represents an assignment to each species of a number, the corresponding molar concentration. Readers who wish to do so can, in this case, simply regard $c$ to be the 3-vector $[c_{NO}, c_{O_2}, c_{NO_2}]$, with the understanding that the arrangement of the three numbers in such an \emph{ordered} array is superfluous to the mathematics at hand.

	Vector representations in spaces such as $\mathbb{R}^I$ rather than  $\mathbb{R}^N$ have advantages in consideration of graphs and networks, where one wants to avoid nomenclature that imparts an artificial numerical order to vertices or edges. See, for example, \cite{Biggs1993}.  

The subset of $\mathbb{R}^I$ consisting of vectors having only positive (nonnegative) components is denoted $\mathbb{R}_+^I\: (\overline{\mathbb{R}}_+^I)$. For each $x \in \mathbb{R}^I$ and for each $z \in \overline{\mathbb{R}}_+^I$, the symbol $x^z$ denotes the real number defined by \[x^z := \prod_{i \in I} (x_{i})^{z_i}, \] where it is understood that $0^0 = 1.$ For each $x, x' \in \mathbb{R}^I$, the symbol $x \circ x'$ denotes the vector of $\mathbb{R}^I$ defined by \[ (x \circ x')_i := x_i x'_i, \hspace{0.2 cm} \forall i \in I.\] 
\noindent
	For each $i \in I$, we denote by $\omega_i$ the vector of $\mathbb{R}^I$ such that $(\omega_i)_j = 1$ whenever $j=i$ and $(\omega_i)_j = 0$ whenever $j \neq i$.

The \emph{standard basis} for $\mathbb{R}^I$ is the set $\left\{\omega_i \in \mathbb{R}^I : i \in I\right\}$. Thus, for each $x \in \mathbb{R}^I$, we can write $x = \sum_{i \in I} x_i \omega_i$. The \emph{standard scalar product} in $\mathbb{R}^I$ is defined as follows: If $x$ and $x'$ are  elements  of $\mathbb{R}^I$, then \[ x \cdot x' = \sum_{i \in I} x_i x'_i .\] The standard basis of $\mathbb{R}^I$ is orthonormal with respect to the standard scalar product. It will be understood that $\mathbb{R}^I$ carries the standard scalar product and the norm derived from the standard scalar product. It will also be understood that $\mathbb{R}^I$ carries the corresponding norm topology. 

If $U$ is a linear subspace of $\mathbb{R}^I$, we denote by $U^\perp$ the orthogonal complement of $U$ in $\mathbb{R}^I$ with respect to the standard scalar product.

By the \emph{support} of $x \in \mathbb{R}^I$, denoted $\supp x$, we mean the set of indices $i \in I$ for which $x_i$ is different from zero. When $\xi$ is a real number, the symbol $\sgn(\xi)$ denotes the sign of $\xi$. When $x$ is a vector of $\mathbb{R}^I$, $\sgn(x)$ denotes the function with domain $I$ defined by \[(\sgn(x))_i := \sgn x_i, \hspace{0.2 cm} \forall i \in I. \]

\subsection{Some definitions}
\label{subsec:Definitions}

	As indicated in Section \ref{sec:MainTheorems}, the objects in a reaction network that appear at the heads and tails of the reaction arrows are the \emph{complexes} of the network. Thus, in  network \eqref{EQ:ExampleNetwork}, the complexes are $2A, B, C, C+D$ and $E$.  A reaction network can then be viewed as a directed graph, with complexes playing the role of the vertices and reaction arrows playing the role of the edges.
\begin{align}
\label{EQ:ExampleNetwork}
\nonumber &\hspace{-2mm}2A \hspace{2 mm} \rightleftarrows \hspace{1.5 mm} B \\
\nonumber & \hspace{.5 mm} \nwarrow \hspace {3 mm}  \swarrow \\
& \hspace{5 mm} C \\[0.75 em]
\nonumber & C + D \rightleftarrows E
\end{align}

\begin{rem}
Let  \scrS \; be the set of species in a network. In chemical reaction network theory, it is sometimes the custom to replace symbols for the standard basis of \RS  \  with the names of the species themselves. For example, in network \eqref{EQ:ExampleNetwork}, with $\scrS = \{A,B,C,D,E\}$,  a vector such as $\omega_C +\omega_D \in \RS$ can instead be written as $C+D$, and $2\omega_A$ can be written as $2A$. In this way, \RS\  can be identified with the vector space of formal linear combinations of the species.  As a result, the complexes of a reaction network with species set \scrS \;can be identified with vectors in \RS.

\end{rem}

\begin{definition}
\label{DEF:ChemicalReactionNetwork}
A \emph{chemical reaction network} consists of three finite sets:
\begin{enumerate}
	\item a set $\mathscr{S}$ of distinct \emph{species} of the network;
	\item a set $\mathscr{C} \subset \bar{\mathbb{R}}_+^{\mathscr{S}}$ of distinct \emph{complexes} of the network;
	\item a set $\mathscr{R} \subset \mathscr{C} \times \mathscr{C}$ of distinct \emph{reactions}, with the following properties:
		\begin{enumerate}
			\item $(y,y) \notin \mathscr{R}$ for any $y\in \mathscr{C}$;
			\item for each $y \in \mathscr{C}$ there exists $y' \in \mathscr{C}$ such that $(y,y') \in \mathscr{R}$ or such that $(y',y) \in \mathscr{R}$.
		\end{enumerate}
	\end{enumerate}
\end{definition}
\smallskip

If $(y,y')$ is a member of the reaction set $\mathscr{R}$, we say that $y$ \emph{reacts to} $y'$, and we write $y \rightarrow y'$ to indicate the reaction whereby complex $y$ reacts to complex $y'$. The complex situated at the tail of a reaction arrow is the \emph{reactant complex} of the corresponding reaction, and the complex situated at the head is the reaction's \emph{product complex}.
\smallskip

	For network \eqref{EQ:ExampleNetwork}, $\scrS = \{A, B, C, D, E\}$, $\scrC = \{2A, B, C, C + D, E\}$, and $\scrR = \{2A \to B, B \to 2A, B \to C, C \to 2A, C + D \to E, E \to C + D\}$. The diagram in \eqref{EQ:ExampleNetwork} is an example of a \emph{standard reaction diagram}: each complex in the network is displayed precisely once, and each reaction in the network is indicated by an arrow in the obvious way. 

	In the context of the present paper and its predecessor \cite{shinar_concordant_2011} the idea of \emph{weak reversibility} plays an important role. The following definition provides some preparation.    

\begin{definition}
\label{DEF:UltimatelyReactsTo}
A complex $y \in \mathscr{C}$ \emph{ultimately reacts} to a complex $y' \in \mathscr{C}$ if any of the following conditions is satisfied:
\begin{enumerate}
\item{$y \rightarrow y' \in \mathscr{R}$};
\item{There is a sequence of complexes $y(1), y(2), \ldots , y(k)$ such that \[ y \rightarrow y(1) \rightarrow y(2) \rightarrow \ldots \rightarrow y(k) \rightarrow y'.\]}
\end{enumerate}
\end{definition}
\smallskip

In network \eqref{EQ:ExampleNetwork} the complex $2A$ ultimately reacts to the complex $C$, but the complex $C$ does not ultimately react to the complex $C + D$. 

\begin{definition}
\label{EQ:WeaklyReversible}
 A reaction network \rnet\ is \emph{reversible} if $y' \rightarrow y \in \mathscr{R}$ whenever $y \rightarrow y' \in \mathscr{R}$. The network is \emph{weakly reversible} if for each $y, y' \in \mathscr{C}$, $y' \mbox{ ultimately reacts to } y$ whenever $y \mbox{ ultimately reacts to } y'$.
\end{definition}

Network \eqref{EQ:ExampleNetwork} is weakly reversible but not reversible. On the other hand, every  reversible  network is also weakly reversible. A reaction network is weakly reversible if and only if in its standard reaction diagram every arrow resides in a directed arrow-cycle.
\smallskip

\begin{definition}
\label{DEF:Reaction vectors}
The \emph{reaction vectors} for a reaction network \rnet\  are the members of the set \[ \left\{y' - y \in \RS : y \rightarrow y' \in \mathscr{R} \right\}. \]
The \emph{rank} of a reaction network is the rank of its set of reaction vectors.
\end{definition}
\smallskip

For network \eqref{EQ:ExampleNetwork} the reaction vector corresponding to the reaction $2A \to B$ is $B - 2A$. The reaction vector corresponding to the reaction $C + D \to E$ is $E - C - D$. 

\begin{definition}
\label{DEF:StoichiometricSubspace}
The \emph{stoichiometric subspace} $S$ of a reaction network \rnet \  is the linear subspace of $\RS$ defined by 

\begin{equation}
S := \mathrm{span} \left\{ y' - y \in \RS : y \rightarrow y' \in \mathscr{R} \right\}.
\label{EQ:StoichiometricSubspace}
\end{equation}
\end{definition}
\smallskip

	The dimension of the stoichiometric subspace is identical to the network's rank. The stoichiometric subspace is a proper linear subspace of \RS \, whenever the network is \emph{conservative}:

\begin{definition}
\label{DEF:Conservative}
A reaction network \rnet \ is \emph{conservative} if the orthogonal complement $S^\perp$ of the stoichiometric subspace $S$ contains a strictly positive member of $\RS$: \[S^\perp \cap \PS \neq \emptyset.\]
\end{definition}

Network \eqref{EQ:ExampleNetwork} is conservative: The strictly positive vector $(A + 2B + 2C + D + 3E) \in \PS$ is orthogonal to each of the reaction vectors of \eqref{EQ:ExampleNetwork}.
\bigskip

For a reaction network \rnet\   a mixture state is generally  represented by a \emph{composition vector} $c \in \PbarS$, where, for each $s \in \mathscr{S}$, we understand $c_s$ to be the molar concentration of $s$. By a \emph{positive composition} we mean a strictly positive composition --- that is, a composition in $\PS$.

\begin{definition}
\label{DEF:Kinetics}
A \emph{kinetics} $\mathscr{K}$ for a reaction network \rnet\ is an assignment to each reaction $y \to y' \in \scrR$ of a \emph{rate function} $\mathscr{K}_{y \to y'}:\PbarS \to \overline{\mathbb{R}}_+$ such that \[ \mathscr{K}_{y \to y'}(c) > 0 \mbox{ if and only if } \supp y \subset \supp c.\]
\end{definition}
\smallskip

\begin{definition}
\label{DEF:KineticSystem}
A \emph{kinetic system} \kinsys\  is a reaction network \rnet\  taken with a kinetics $\mathscr{K}$ for the network.
\end{definition}
\smallskip

\medskip
Many of the dynamical consequences of network concordance require that the kinetics be \emph{weakly monotonic} or \emph{differentiably monotonic}. Both are formally defined below:

\begin{definition}
\label{def:WeaklyMonotonic}
A kinetics \scrK\   for reaction network \rnet\  is \emph{weakly monotonic} if, for each pair of compositions $c^*$ and $c^{**}$, the following implications hold for each reaction $\rxn \in \scrR$ such that  $\supp y \subset \supp c^*$ and $\supp y \subset \supp c^{**}$:\\

\noindent (i) $\scrK_{\rxn}(c^{**})\; >\; \scrK_{\rxn}(c^{*})\quad \Rightarrow \quad$ there is a species $s \in \supp y$ with $c_s^{**} > c_s^*$.\\

\noindent (ii) $\scrK_{\rxn}(c^{**})\;=\; \scrK_{\rxn}(c^{*})\quad \Rightarrow \quad$ $c_s^{**} = c_s^*$ for all $s \in \supp y$ or else there are species $s, s' \in \supp y$ with $c_s^{**} > c_s^*$ and $c_{s'}^{**} < c_{s'}^*$.\\

\noindent We say that the kinetic system \kinsys\ is weakly monotonic when its kinetics \scrK\, is weakly monotonic.
\end{definition}

\begin{definition}
\label{def:DiffMonotonic}
 A kinetics \scrK \, for a reaction network \rnet\,is \emph{differentiably monotonic} at $c^* \in \PS$ if, for every  reaction $\rxn \in \scrR$, $\scrK_{\rxn}(\cdot)$ is differentiable at $c^*$ and, moreover, for each species $s \in \scrS$,
\begin{equation}
\frac{\partial \scrK_{\rxn}}{\partial c_s}(c^*) \ge 0,
\end{equation}
with inequality holding if and only if $s \in \supp y$. A \emph{differentiably monotonic kinetics} is one that is differentiably monotonic at every positive composition.
\end{definition}

	When a kinetics $\mathscr{K}$ for a reaction network \rnet is differentiably monotonic at $c^* \in \PS$, we denote  by the symbol $\nabla \mathscr{K}_{y \to y'} (c^*)$ the member of $\RS$ defined by
\[
\nabla \mathscr{K}_{y \to y'} (c^*) := \sum_{s \in \mathscr{S}} \frac{\partial{\mathscr{K}_{y \to y'} }}{\partial{c_s}}(c^*)s.
\]

	Note that every mass action kinetics is both weakly monotonic and differentiably monotonic. Recall that a \emph{mass action kinetics} is a kinetics in which the rate function corresponding to each reaction $y \to y'$ takes the form $\scrK_{y \to y'}(c) = k_{y \to y'} c^y, \forall c \in \PbarS$, where $k_{y \to y'}$ is a positive \emph{rate constant} for the reaction \rxn.  

\begin{definition}
\label{DEF:SpeciesFormationRateFunction}
The \emph{species formation rate function} for a kinetic system \kinsys\  with stoichiometric subspace $S$ is the map $f:\PbarS \to S$ defined by

\begin{equation}
\label{EQ:SpeciesFormationRateFunction}
f(c) = \sum_{\rxn \in \scrR} \scrK_{y \to y'} (c) (y' - y). 
\end{equation}
\end{definition}
\smallskip

\begin{definition}
\label{DEF:DifferentialEquation}
The \emph{differential equation} for a kinetic system with species formation rate function $f(\cdot)$ is given by

\begin{equation}
\label{EQ:DifferentialEquation}
\dot{c} = f(c).
\end{equation} 
\end{definition}
\smallskip

	From equations \eqref{EQ:StoichiometricSubspace}, \eqref{EQ:SpeciesFormationRateFunction}, and \eqref{EQ:DifferentialEquation}, it follows that, for a kinetic system \kinsys,  the vector $\dot{c}$ will invariably lie in the stoichiometric subspace $S$ for the network \rnet. Thus,  two compositions $c \in \PbarS$ and $c' \in \PbarS$ can lie  along   the same solution of (\ref{EQ:DifferentialEquation}) only if their difference $c' - c$ lies in $S$. This motivates the following definition: 

\begin{definition}
\label{DEF:StoichiometricallyCompatible}
Let \rnet\  be a reaction network with stoichiometric subspace $S$. Two compositions $c$ and $c'$ in $\PbarS$ are \emph{stoichiometrically compatible} if $c' - c$ lies in $S$.
\end{definition}
\smallskip

	For a network \rnet, the stoichiometric compatibility relation serves to partition $\PbarS$ into equivalence classes  called the  \emph{stoichiometric compatibility classes} for the network.   Thus, the stoichiometric compatibility class containing an arbitrary composition $c$, denoted $(c + S) \cap \PbarS$, is given by

\begin{equation}
\label{EQ:StoichiometricCompatibilityClass}
(c + S) \cap \PbarS = \left\{ c' \in \PbarS: c' - c \in S \right\}.
\end{equation}
The notation is intended to suggest that $(c + S) \cap \PbarS$ is the intersection of $\PbarS$ with the parallel of $S$ containing $c$.  A stoichiometric compatibility class is \emph{nontrivial} if it contains at least one (strictly) positive composition.

\begin{definition}
\label{DEF:Equilibrium}
\noindent An \emph{equilibrium} of a kinetic system \kinsys\  is a composition $c \in \PbarS$ for which $f(c) = 0$. A \emph{positive equilibrium}  is an equilibrium that lies in $\PS$.
\end{definition}
\smallskip

	Because compositions along solutions of \eqref{EQ:DifferentialEquation} are stoichiometrically compatible, one is typically interested in changes in values of the species formation rate function that result from changes in composition that are stoichiometrically compatible.  In particular, if $f(c^*)$ is the value of the species formation rate function at composition $c^*$, then one might be interested in the value of $f(c)$ for a composition $c$ very close to  $c^*$ and  stoichiometrically compatible with it. Thus, for a kinetic system \kinsys\ with stoichiometric subspace $S$, with smooth reaction rate functions, and with species formation rate function  $f:\PbarS \to S$, we will want to work with  the derivative  $df(c^*): S \to S$, defined by

\begin{equation}
\label{EQ:Derivative}
df(c^*)\sigma = \left. \frac{df(c^* + \theta \sigma)}{d\theta} \right|_{\theta = 0}, \quad  \forall \sigma \in S.
\end{equation}

\medskip 
	We say that $c^* \in \PS$ is a \emph{nondegenerate equilibrium} if $c^*$ is an equilibrium and if, moreover, $df(c^*)$ is nonsingular.  An \emph{eigenvalue associated with a positive equilibrium} $c^* \in \PS$ is an eigenvalue of the derivative $df(c^*)$.

\section{Concordant reaction networks}
\label{sec:ConcordantNetworks}
Here we recall the definition of reaction network concordance \cite{shinar_concordant_2011}.  We consider a reaction network \rnet with stoichiometric subspace $S \subset \RS$, and we let $L: \RR \to S$ be the linear map defined by
\begin{equation}
L\alpha = \sum_{\rxn \in \scrR}\alpha_{\rxn}(y' - y).
\end{equation}

\begin{definition} The reaction network \rnet is \emph{concordant} if there do not exist an $\alpha \in \ker L$ and a nonzero $\sigma \in S$ having the following properties:\\

\noindent (i) For each \rxn\; such that $\alpha_{\rxn} \neq 0$, \supp $y$ contains a species $s$ for which $\mathrm{sgn}\, \sigma_s = \mathrm{sgn}\, \alpha_{\rxn}$.\\

\noindent (ii) For each \rxn\; such that $\alpha_{\rxn} = 0$, $\sigma_s = 0$\, for all $s \in \supp y$ or else \supp $y$ contains species $s\, \textrm{and } s'$ for which  $\mathrm{sgn}\, \sigma_s = -\, \mathrm{sgn}\; \sigma_{s'}$, both not zero.
\medskip\noindent\\
A network that is not concordant is \emph{discordant}.
\label{def:Concordant}
\end{definition}

	Note that for a fully open network with species set \scrS\ the stoichiometric subspace coincides with \RS.  The following lemma will prove useful later on: 

\begin{lemma} 
\label{lemma:RemoveRevRxns}
If a fully open network is discordant, it is always possible to choose from each reversible pair of non-degradation reactions at least one (and sometimes both) of the reactions for removal such that the resulting fully open subnetwork is again discordant.
\end{lemma} 

\begin{proof} Suppose that a fully open network \rnet is discordant. Then there are $\alpha \in \ker L$ and nonzero $\sigma$ that together satisfy conditions $(i)$ and $(ii)$ in Definition \ref{def:Concordant}. In particular, we have
\begin{equation}
\label{eq:BasicAlphaEqn}
\sum_{\rxn \in \scrR}\alpha_{\rxn}(y' - y) = 0.
\end{equation}

Let $\bar{y} \rightleftarrows \bar{y}'$ be a pair of reversible non-degradation reactions in \scrR, and suppose that $\alpha_{\bar{y} \to \bar{y}'} \neq \alpha_{\bar{y}' \to \bar{y}}$ with the complexes labeled such that $|\alpha_{\bar{y} \to \bar{y}'}| > |\alpha_{\bar{y}' \to \bar{y}}|$. In this case, let $\scrR^{\dagger} := \scrR \setminus \{\bar{y}' \to \bar{y}\}$, let $\alpha^{\dagger}_{\bar{y} \to \bar{y}'} := \alpha_{\bar{y} \to \bar{y}'} - \alpha_{\bar{y}' \to \bar{y}}$, and, for all other \rxn\ in $\scrR^{\dagger}$, let $\alpha^{\dagger}_{\rxn} := \alpha_{\rxn}.$  From \eqref{eq:BasicAlphaEqn} it follows easily that
\begin{equation}
\label{eq:ReducedAlphaEqn}
\sum_{\rxn \in \scrR^{\dagger}}\alpha^{\dagger}_{\rxn}(y' - y) = 0.
\end{equation}

	If, on the other hand, $\alpha_{\bar{y} \to \bar{y}'} = \alpha_{\bar{y}' \to \bar{y}}$, then we can choose $\scrR^{\dagger} := \scrR \setminus \{\bar{y}' \rightleftarrows \bar{y}\}$, and,  for all  \rxn\ in $\scrR^{\dagger}$, we can let $\alpha^{\dagger}_{\rxn} := \alpha_{\rxn}.$ In this case, \eqref{eq:ReducedAlphaEqn} will again obtain.

	In either case, it is easy to see that $\alpha^{\dagger}$, taken with the original $\sigma$, suffices to establish the discordance of the subnetwork associated with the reaction set $\scrR^{\dagger}$.
\end{proof}

\begin{rem} In effect, Lemma \ref{lemma:RemoveRevRxns} tells us that every fully open discordant network with reversible non-degradation reactions has a discordant fully open subnetwork in which no non-degradation reaction is reversible. Note that, apart from minor changes in labels within the reaction nodes (i.e., replacement of $y \rightleftarrows y'$ by \rxn) the Species-Reaction Graph derived from the indicated discordant subnetwork is a subgraph of the Species-Reaction Graph drawn for the original network. As we explain at the beginning of Section \ref{sec:ProofOfConcordanceThm}, that subgraph will satisfy the conditions of Theorem \ref{thm:ConcordanceTheorem} and its corollary if the parent Species-Reaction Graph does. These observations will help us simplify certain arguments that are otherwise complicated by the presence of reversible reactions. 
\end{rem}

\medskip

 \section{Proof of Theorem \ref{thm:ConcordanceTheorem}}
\label{sec:ProofOfConcordanceThm}
The proof will be by contradiction. With this in mind, we assume hereafter the true-chemistry Species-Reaction Graph (SR Graph) for the fully open network under consideration has both attributes of the theorem statement and that, contrary to the assertion of the theorem, the fully open network is discordant.  

	In this case, Lemma \ref{lemma:RemoveRevRxns} tells us that, when the true chemistry has reversible reaction pairs, then each such pair can be replaced by a certain irreversible reaction of the pair (or else removed completely) such that the resulting fully open network is again discordant. The SR Graph for the altered (totally irreversible) true chemistry might have fewer cycles than in the original SR Graph (but never more) because of removal of reversible reaction pairs. Similarly, there might be fewer \emph{orientable} cycles (but never more) than in the original SR Graph as a result of replacement of reversible reaction pairs by single irreversible reactions. Moreover, there might be fewer critical subgraphs (but never more) than in the original SR Graph because of a loss of orientable cycles or because more cycles have acquired compulsory orientations.  (Cycles that were even in the original SR Graph and persist in the new SR Graph remain even.) For all of these reasons, \emph{the SR Graph for the totally irreversible subnetwork of the original ``true" chemical reaction network will, like the SR Graph for the original network, satisfy the requirements of Theorem \ref{thm:ConcordanceTheorem}}.

	This is to say that there is no loss of generality in assuming, for the purposes of contradiction, that there is a ``true" chemical reaction network, containing no reversible reactions, for which the SR Graph satisfies both requirements of Theorem \ref{thm:ConcordanceTheorem} but for which the fully open extension is discordant.

	Hereafter in the proof of Theorem \ref{thm:ConcordanceTheorem}, then, \emph{we assume that all reactions in the ``true" chemistry are irreversible, that the corresponding SR Graph satisfies both conditions of  Theorem \ref{thm:ConcordanceTheorem}, and that, contrary to what the theorem asserts,  the fully open extension of the true chemistry is discordant}. 

	Thus there exist fixed $\alpha \in \ker L$ and nonzero $\sigma \in \RS$ satisfying the requirements of Definition \ref{def:Concordant}. (Recall that the stoichiometric subspace for a fully open network is \RS.) \emph{It will be understood throughout the proof that all references to $\alpha$ and $\sigma$ are relative to this fixed pair, so chosen}. 	
	
	It will be helpful to divide the proof into subsections:

\subsection{Preliminaries} We associate a sign with each species in the following way:  \emph{A species $s \in \scrS$ is positive, negative, or zero} according to whether $\sigma_s$ is positive, negative, or zero. By a \emph{signed species} we mean one that is either positive or negative. Similarly, \emph{a reaction $\rxn \in \scrR$ is positive, negative, or zero} according to whether $\alpha_{\rxn}$ is positive, negative, or zero. A \emph{signed reaction} is one that is either positive or negative.

\begin{rem}
\label{rem:signremark}
	Note in particular, that for any ``degradation reaction" of the kind $s \to 0$, Definition \ref{def:Concordant} requires that sgn $\alpha_{s \to 0}$ = sgn $\sigma_s$, so sgn $s \to 0$ = sgn $s$ for every $s \in \scrS$. (A similar situation obtains for any reaction of the form $ns \to y$, where $n$ is a positive number.) For any ``synthesis reaction" $0 \to s$, Definition \ref{def:Concordant} requires that $\alpha_{0 \to s} = 0$, so such reactions are unsigned.
\end{rem}

\subsection{The sign-causality graph and causal units}
\label{sec:SignCausalityGraph}
 Hereafter we denote by  \scrRst\,  the set of all reactions not of the form \mbox{$s \to 0,\ s \in \scrS$}. That is, \scrRst\, is the set of reactions that are not degradation reactions. Because $\alpha$ is a member of $\ker L$ we can write
\begin{equation}
\sum_{\rxn \in \scrRst}\alpha_{\rxn}(y' - y) - \sum_{s \in \scrS}\alpha_{s \to 0}s = 0.
\end{equation}
Thus, for any particular choice of species $s \in \scrS$, we have
\begin{equation}
\label{eq:alphaeqn}
\sum_{\rxn \in \scrRst}\alpha_{\rxn}(y'_s - y_s) = \alpha_{s \to 0}.
\end{equation}

	Now \emph{suppose that in \eqref{eq:alphaeqn}, the species $s$ is positive.}  From \eqref{eq:alphaeqn} and Remark \ref{rem:signremark} we have
\begin{equation}
\label{eq:alphaeqn>}
\sum_{\rxn \in \scrRst}\alpha_{\rxn}(y'_s - y_s) > 0,
\end{equation}
in which case at least one term on the left must be positive. (At least one term on the left is ``causal" for the inequality.\footnote{When we say that X is causal for the outcome Y, we mean to suggest that X abets Y, not necessarily that X, by itself, determines the outcome Y.}) Terms of this kind can arise in precisely two ways:
\medskip 

\noindent (i) There is a \emph{positive} reaction \rxn\; (i.e., $\alpha_{\rxn}$ is positive) with \mbox{$s \in \supp y'$} (so that $y'_{s} > 0, y_s = 0$). Recall, however, that for $\alpha_{\rxn}$ to be positive, the conditions of Definition \ref{def:Concordant} require that there be a \emph{positive} species $s'$ in $\supp y$.  

	In this case, reaction \rxn \, is ``causal" for the sign of species $s$, while species $s'$ is ``causal" for the sign of reaction \rxn. With this in mind, we write
\begin{equation}
\label{eq:pluspluscausal}
\overset{+}{s'} \quad \overset{y}{\rightsquigarrow} \quad \overset{+}{\framebox{\rxn}} \quad \overset{y'}{\rightsquigarrow}\quad \overset{+}{s}.
\end{equation}

\medskip
\noindent The signs above the species and the reactions indicate their respective signs. The complex labels above the ``causal" arrows ($\rightsquigarrow$) indicate the complex in whose support the adjacent species resides.

\bigskip
\noindent (ii) There is a \emph{negative} reaction \rxn\; (i.e., $\alpha_{\rxn}$ is negative) with \mbox{$s \in \supp y$} (so that $y_{s} > 0, y'_s = 0$). But for $\alpha_{\rxn}$ to be negative, the conditions of Definition \ref{def:Concordant} require that there be a \emph{negative} species $s'$ in $\supp y$.
	As in case (i), reaction \rxn \, is causal for the sign of species $s$, while species $s'$ is causal for the sign of reaction \rxn. This time we write 
\begin{equation}
\label{eq:minuspluscausal}
\overset{-}{s'} \quad \overset{y}{\rightsquigarrow} \quad \overset{-}{\framebox{\rxn}} \quad \overset{y}{\rightsquigarrow}\quad \overset{+}{s}.
\end{equation}

	On the other hand,  \emph{suppose that in \eqref{eq:alphaeqn}, the species $s$ is negative.}  From \eqref{eq:alphaeqn} and Remark \ref{rem:signremark} we have
\begin{equation}
\label{eq:alphaeqn<}
\sum_{\rxn \in \scrRst}\alpha_{\rxn}(y'_s - y_s) < 0,
\end{equation}
in which case at least one term on the left must be negative. (At least one term on the left is causal for the inequality.) Here again, terms of this kind can arise in precisely two ways:
\medskip

\noindent (i)${}^\prime$ \,There is a \emph{negative} reaction \rxn\; (i.e., $\alpha_{\rxn}$ is negative) with \mbox{$s \in \supp y'$} (so that $y'_{s} > 0, y_s = 0$). For $\alpha_{\rxn}$ to be negative, the conditions of Definition \ref{def:Concordant} require that there be a \emph{negative} species $s'$ in $\supp y$.  Here we write

\begin{equation}
\label{eq:minusminuscausal}
\overset{-}{s'} \quad \overset{y}{\rightsquigarrow} \quad \overset{-}{\framebox{\rxn}} \quad \overset{y'}{\rightsquigarrow}\quad \overset{-}{s}.
\end{equation}
\smallskip

\noindent (ii)${}^\prime$\, There is a \emph{positive} reaction \rxn\; (i.e., $\alpha_{\rxn}$ is positive) with \mbox{$s \in \supp y$} (so that $y_{s} > 0, y'_s = 0$). For $\alpha_{\rxn}$ to be positive, the conditions of Definition \ref{def:Concordant} require that there be a \emph{positive} species $s'$ in $\supp y$. We write

\begin{equation}
\label{eq:plusminuscausal}
\overset{+}{s'} \quad \overset{y}{\rightsquigarrow} \quad \overset{+}{\framebox{\rxn}} \quad \overset{y}{\rightsquigarrow}\quad \overset{-}{s}.
\end{equation}
\smallskip

\noindent	The diagrams drawn in \eqref{eq:pluspluscausal}, \eqref{eq:minuspluscausal}, \eqref{eq:minusminuscausal} and \eqref{eq:plusminuscausal} can be viewed as edge-pairs in a bipartite directed graph:
\smallskip

	The \emph{sign-causality graph}, drawn for the network (relative to the $\alpha$, $\sigma$ pair under consideration) is constructed according to the following prescription: The vertices  are the signed species and signed (non-degradation) reactions. An edge $\leadsto$ is drawn from a signed species $s'$ to a signed reaction \rxn\ whenever $s'$ is contained in  $\supp y$ and the two signs agree; the edge is then labeled with the complex $y$. An edge $\leadsto$ is drawn from a signed reaction \rxn\, to a signed species $s$ in either of the following situations: (i) $s$ is contained in $\supp y'$ and the sign of $s$ agrees with the sign of the reaction; in this case the edge carries the label $y'$ or (ii) $s$ is contained in $\supp y$ and the sign of $s$ disagrees with the sign of the reaction; in this case the edge carries the label $y$.  It is understood that the signed species and the signed reactions are labeled by their corresponding signs.

	By a \emph{causal unit} we mean a directed two-edge subgraph of the sign-causality graph of the kind $s' \leadsto R \leadsto s$, where $R$ denotes a reaction. (We will often designate a reaction by the symbol $R$ when there is no need to indicate the reactant and product complexes.) The species $s'$ is the \emph{initiator} of the causal unit $s' \leadsto R \leadsto s$, while $s$ is its \emph{terminator}. It is not difficult to see that the initiator and terminator of a causal unit must be \emph{distinct} species.\footnote{Recall that, by supposition, a species can appear on only one side of a reaction.}

	Causal units are of the four varieties shown in \eqref{eq:pluspluscausal}, \eqref{eq:minuspluscausal}, \eqref{eq:minusminuscausal} and \eqref{eq:plusminuscausal}. Of these, \eqref{eq:pluspluscausal} and \eqref{eq:minusminuscausal} carry distinct complex labels on the two edges. On the other hand, \eqref{eq:minuspluscausal} and \eqref{eq:plusminuscausal} carry identical labels on the edges. By a \emph{c-pair causal unit} we mean a causal unit of the second kind. (As with the Species-Reaction Graph, the term is meant to be mnemonic for \emph{complex pair}.)

\begin{rem} 
\label{rem:SignChangeCPair}
It is important to note that a c-pair causal unit results in a \emph{change} of sign as the edges are traversed from the initiator species to the terminator species. Otherwise, a causal unit is characterized by \emph{retention} of the sign.
\end{rem}

\begin{rem}
\label{rem:underlyingSRGraph}
It should be clear that, apart from the direction $\leadsto$ imparted to its edges, the sign-causality graph corresponding to $\alpha$ and $\sigma$ can be identified with a subgraph of the undirected Species-Reaction Graph (which we will sometimes refer to as the sign-causality graph's ``counterpart" in the Species-Reaction Graph). Moreover, every fixed-direction edge-pair in the SR Graph has direction consistent with that imparted by the $\leadsto$-relation (because each proceeds from a reactant species to an irreversible reaction to a product species).
\end{rem}

\subsection{The sign-causality graph must contain a directed cycle, and all of its directed cycles are even.}
\label{subsec:EvenCycleSec}
 By supposition $\sigma$ is not zero, so there is at least one signed species, say $s_1$. From the discussion in Section \ref{sec:SignCausalityGraph} it is clear that  $s_1$ is the terminator of a causal unit $s_2 \leadsto R_2 \leadsto s_1$, where $s_2$ is distinct from the initiator $s_1$. Because $s_2$ is also signed, it too must be the terminator of a causal unit $s_3 \leadsto R_3 \leadsto s_2$, where the signed species $s_3$ is distinct from $s_2$.  Continuing in this way, we can see that there is a directed sequence of the form
\begin{equation}
\cdots \leadsto s_n \leadsto R_n \leadsto s_{n-1} \leadsto R_{n-1} \cdots \leadsto  s_3 \leadsto R_3 \leadsto s_2 \leadsto R_2 \leadsto s_1, 
\end{equation}
with $s_i \neq s_{i+1}$.  

	Because the number of species is finite, two non-consecutive species in the sequence must in fact coincide, which is to say that \emph{the sign-causality graph must contain a directed cycle}.  Moreover, it is easy to see that each vertex in the sign-causality graph resides in a cycle or else there is a cycle $\leadsto$-upstream from it.

	 With this in mind, we record here some vocabulary that will be useful in the next section: A \emph{source} is a strong component of the sign-causality graph whose vertices have no incoming edges originating at vertices outside that strong component. Clearly, every component of the sign-causality graph contains a source, and every vertex in a source resides in a directed cycle.

	As with the SR Graph, we say that a (not necessarily directed) cycle in the sign-causality graph is \emph{even} if it contains an even number of c-pairs. 

\begin{lemma}
\label{lem:EvenDirectedCycle}
 A (not necessarily directed) cycle in the sign-causality graph that is the union of causal units is even. In particular, every directed cycle in the sign-causality graph is even.
\end{lemma}

\begin{proof}
If we traverse the cycle beginning at a species $s^0$ and count the number of species-sign changes when we have returned to $s^0$, that number clearly must be even. But, if the cycle is the union of causal units, the number of sign changes is identical to the number of  c-pair causal units the cycle contains (Remark \ref{rem:SignChangeCPair}). Clearly, every directed cycle in the sign-causality graph is the union of causal units.
\end{proof}

	Because a directed cycle in the sign-causality graph is even, its (orienttable) cycle counterpart in the undirected Species-Reaction Graph (Remark \ref{rem:underlyingSRGraph}) will also have an even number of c-pairs. Because the sign-causality graph must contain a source, which in turn must contain a directed cycle, we now know that \emph{a reaction network is concordant if its Species-Reaction Graph contains no orientable even cycles}. 
\medskip
\begin{rem}
\label{rem:SourceAsCriticalSubgraph}
In fact, \emph{a source in the sign-causality graph, when viewed in the SR Graph, must be a critical subgraph}. That this is so follows from the fact that a source is a strong component of the sign-causality graph and therefore is the union of $\leadsto$-directed cycles. These even cycles, viewed in the SR Graph, have a consistent orientation, the orientation conferred by the directed-cycle $\leadsto$--orientation in the sign-causality graph, which in turn is consistent with any fixed-direction edge-pairs the SR Graph might contain.
\end{rem}
\medskip

	In the next section we begin to consider what happens when the Species-Reaction Graph does contain at least one orientable even cycle (and therefore at least one critical subgraph). We will want to show that if the fully open network under consideration satisfies the hypothesis of Theorem \ref{thm:ConcordanceTheorem} then the very existence of a source in the putative sign-causality graph becomes impossible.\footnote{As a matter of vocabulary, a strong component of a directed graph --- in particular a source --- is, formally, a set of vertices in the graph, but here, less formally, we will also associate the source with the subgraph of the sign-causality graph obtained by joining the vertices of the source with directed edges inherited from the parent graph.}
\bigskip

\subsection{Inequalities associated with a source; stoichiometric coefficients associated with edges in the sign-causality graph} Consider a source in the sign-causality graph having $\scrSo \subset \scrS$ as its set of species nodes and $\scrRo \subset \scrRst$ as its set of reaction nodes. If $s$ is a \emph{positive} species in the source, then  \eqref{eq:alphaeqn>} can be written

\begin{equation}
\label{eq:alphaeqn>two-part}
\sum_{\rxn \in \scrRo}\alpha_{\rxn}(y'_s - y_s) \;\;+ \sum_{\rxn \in \scrRst \setminus \scrRo}\alpha_{\rxn}(y'_s - y_s)> 0.
\end{equation}
\noindent
Now suppose that a term in the second sum on the left, corresponding to reaction $\bar{y} \to \bar{y}'$, is not zero. Because $\bar{y} \to \bar{y}'$ is not a vertex of the source, any edge of the sign-causality graph that connects species $s$ to reaction $\bar{y} \to \bar{y}'$ must point \emph{away} from $s$. Thus, the reaction $\bar{y} \to \bar{y}'$ cannot be causal for the positive sign of species $s$, so the term 
\begin{equation}
\alpha_{\bar{y} \to \bar{y}'}(\bar{y}'_s - \bar{y}_s) \nonumber
\end{equation}
must be negative. This implies that \eqref{eq:alphaeqn>two-part} can obtain only if we have 

\begin{equation}
\label{eq:SourceInequalityPositive}
\sum_{\rxn \in \scrRo}\alpha_{\rxn}(y'_s - y_s) > 0.
\end{equation}
When $s$ is a \emph{negative} species in the source, we can reason similarly to write
\begin{equation}
\label{eq:SourceInequalityNegative}
\sum_{\rxn \in \scrRo}\alpha_{\rxn}(y'_s - y_s) < 0.
\end{equation}

\begin{rem} 
\label{rem:OutOfSourceEdges}
Note that, for a particular $s \in \scrSo$ a nonzero term in \eqref{eq:SourceInequalityPositive} or \eqref{eq:SourceInequalityNegative} might not correspond to any edge of the sign-causality graph (as when, for a particular reaction \rxn,  $s$ is a member of $\supp y'$ and disagrees in sign with $\alpha_{\rxn}$). If $s$ is a positive species, then the term in question must be negative, and hence it can be removed from \eqref{eq:SourceInequalityPositive} without changing the sense of that inequality. Similarly, if  $s$ is a negative species, the term in question is positive and can be removed from \eqref{eq:SourceInequalityNegative} without changing the inequality's sense. \emph{In what follows below, we shall assume that such removals have been made, so that every term in \eqref{eq:SourceInequalityPositive} or \eqref{eq:SourceInequalityNegative} corresponds to an edge in the sign-causality graph.}
\end{rem}
\bigskip

	Recall that a directed edge of the sign-causality graph from a species $s$ to a reaction \rxn\  is always of the form
\begin{equation}
s \quad \overset{y}{\rightsquigarrow} \quad \framebox{\rxn}. \nonumber
\end{equation}
That is, the edge-label of such a species-to-reaction edge is invariably the reactant complex, $y$. Note that species $s$ has a positive stoichiometric coefficient, $y_s$, in that complex. On the other hand, reaction-to-species edges of the sign-causality graph are of two kinds:
\begin{equation}
\framebox{\rxn} \quad \overset{y}{\rightsquigarrow}\quad {s} \quad\quad \textrm{or} \quad \quad \framebox{\rxn} \quad \overset{y'}{\rightsquigarrow}\quad {s}. \nonumber
\end{equation}
In either case, the species $s$ has a positive stoichiometric coefficient (either $y_s$ or $y'_s$) in the edge-labeling complex.

	\emph{Hereafter, for a species-to-reaction edge $s \rightsquigarrow R$ of the sign-causality graph we denote by $e_{s \rightsquigarrow R}$ the (positive) stoichiometric coefficient of species $s$ in the corresponding edge-labeling complex. For a reaction-to-species edge $R \rightsquigarrow s$ we denote by $f_{R \rightsquigarrow s}$ the (positive) stoichiometric coefficient of species $s$ in its edge-labeling complex.}

	For the sign-causality graph source under consideration, we will in fact need a small amount of additional notation: For each species $s \in \scrSo$ we denote by $\scrRo \rightsquigarrow s$ the set of all edges of the source that are incoming to $s$ and by $s \rightsquigarrow \scrRo$ the set of all edges of the source that are outgoing from $s$. In light of Remark \ref{rem:OutOfSourceEdges} and in view of notation we now have available, some analysis will indicate that the inequalities given by \eqref{eq:SourceInequalityPositive} and \eqref{eq:SourceInequalityNegative} can be written as a single system:
\begin{equation}
\label{eq:SourceInequalitySystem}
\boxed{\sum_{\scrR_0 \rightsquigarrow s}f_{R \rightsquigarrow s}|\alpha_R|\; - \sum_{s \rightsquigarrow \scrRo}e_{s \rightsquigarrow R}|\alpha_R| > 0, \quad \forall s \in \scrSo.}
\end{equation}

Our aim will be to show that when the conditions of Theorem \ref{thm:ConcordanceTheorem} obtain, the inequality system \eqref{eq:SourceInequalitySystem} cannot be satisfied. The following remark will play an important role.

\bigskip
\begin{rem}
\label{rem:StoichExpansiveRemark}
 Consider a source in the sign-causality graph, and suppose that

\begin{equation}
\label{eq:SourceDirectedCycle}
s_1 \rightsquigarrow R_1  \rightsquigarrow s_2 \rightsquigarrow R_2 \dots\rightsquigarrow s_{n} \rightsquigarrow R_n \rightsquigarrow s_1
\end{equation}
is one of its  $\rightsquigarrow$-directed (and therefore even) cycles. Note that the  \mbox{$\rightsquigarrow$-orientation} of the cycle also provides an orientation of the cycle's counterpart in the Species-Reaction Graph, for it is consistent with directions of the fixed-direction edge-pairs in the SR Graph. (Remark \ref{rem:underlyingSRGraph})

	\emph{Thus, when condition (i) of Theorem \ref{thm:ConcordanceTheorem} is satisfied, the stoichiometric coefficients of the directed cycle \eqref{eq:SourceDirectedCycle} must satisfy the condition}

\begin{equation}
\label{eq:NotStoichExpansiveCondition}
\frac{ f_{R_1 \rightsquigarrow s_2}f_{R_2 \rightsquigarrow R_3} \dots f_{R_n \rightsquigarrow s_1}}{e_{s_1 \rightsquigarrow R_1}e_{s_2 \rightsquigarrow R_2} \dots e_{s_n \rightsquigarrow R_n}} \leq 1.
\end{equation}

\end{rem}

\bigskip

\subsection{The decomposition of a source into its blocks; the block-tree of a source} Here we draw on graph theory vocabulary that is more-or-less standard \cite{bondy_graph_2010}.\footnote{Our focus is on bipartite graphs, so in the discussion of terminology it will be understood that no vertex is adjacent to itself via a self-loop edge.}   A  \emph{separation} of a connected graph is a decomposition of the graph into two edge-disjoint connected subgraphs, each having at least one edge, such that the two subgraphs have just one vertex in common. If a connected graph admits a separation (in which case it is \emph{separable}), then the common vertex of the separation is called a \emph{separating vertex} of the graph. A graph is \emph{nonseparable} if it is connected and has no separating vertices. A maximal nonseparable subgraph of a graph is called a \emph{block} of the graph. In rough terms, a connected graph is made up of its blocks, pinned together at the graph's separating vertices.

\begin{figure}[ht]
\centering
\includegraphics[scale=.4]{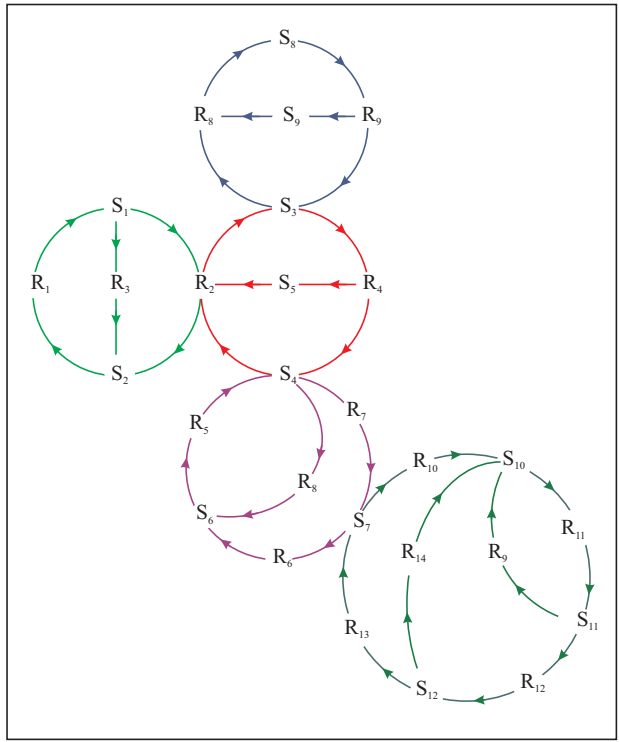}
\caption{A hypothetical source with five  blocks}
\label{fig:BlockIllustration}
\end{figure}

	A source in the sign-causality graph can be decomposed into its blocks, which we call \emph{source-blocks}. Because a source is strongly connected, each of its blocks is strongly connected (and non-separable). To illustrate some of these ideas we show in  Figure \ref{fig:BlockIllustration} a hypothetical source with five blocks. (Although they are inconsequential to the block decomposition, the arrows in the figure are meant to connote the $\rightsquigarrow$-relation.) The separating vertices in the figure are $R_2, S_3,S_4$, and $S_7$. An example of a block is the subgraph having vertices $\{S_1, S_2, R_1, R_2, R_3\}$ and edges $\{S_1 \rightsquigarrow R_2, R_2 \rightsquigarrow S_2, S_2 \rightsquigarrow R_1, R_1 \rightsquigarrow S_1, S_1\rightsquigarrow R_3, R_3 \rightsquigarrow S_2\}$.

	The \emph{block-tree} of a connected graph depicts the way in which the various blocks of the graph are joined at its separating vertices. More precisely, the block-tree \cite{bondy_graph_2010} of a connected graph is a bipartite graph whose vertices are of two kinds: symbols for the graph's blocks and symbols for the graph's separating vertices. If a block contains a particular separating vertex, then an edge is drawn from the block's symbol to the symbol for that separating vertex. In Figure \ref{fig:BlockTree} we show the block-tree for the hypothetical source depicted in Figure \ref{fig:BlockIllustration}. For reasons that will be made clear later on, we have denoted the five blocks in the block-tree by the symbols $RB1,SB1,SB2,RB2$, and $RB3$. (We note that Figures \ref{fig:BlockIllustration} and \ref{fig:BlockTree} are somewhat unrepresentative, for it might happen that three or more source-blocks are adjacent to the same separating vertex.)

\begin{figure}[ht]
\centering
\includegraphics[scale=.35]{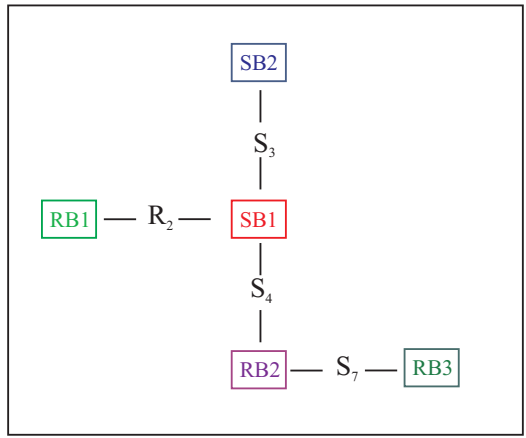}
\caption{The block-tree for Figure \ref{fig:BlockIllustration}}
\label{fig:BlockTree}
\end{figure}

	A block is an \emph{end block} if it contains no more than one separating vertex of the original graph. In this case, the block's symbol is a \emph{leaf} of the block-tree graph. The end-blocks (leaves) in our example are those represented by $RB1, SB2$, and $RB3$.

\subsection{Properties of source-blocks}

	Our aim in this subsection is to show that when condition (ii) of Theorem \ref{thm:ConcordanceTheorem} is satisfied, the internal structure of source-blocks must have a certain degree of simplicity. Condition (ii) will exert itself primarily through the following proposition, which is proved in Appendix \ref{app:ProofOfSourceBlockProposition}.
\bigskip\bigskip

\begin{proposition} 
\label{prop:SourceBlockSimplicity}
Suppose that, for the  reaction network under consideration, the Species-Reaction Graph satisfies condition (ii) of Theorem \ref{thm:ConcordanceTheorem}. Then, in any source-block of the sign-causality graph, at most one of the following can obtain:
\medskip

(i) There is a reaction vertex having more than two adjacent species vertices.
\smallskip

(ii) There is a species vertex having more than two adjacent reaction vertices.
\end{proposition}

\medskip
	The proposition provides motivation for the following definition:

\begin{definition}
\label{def:SBlockRBlock}
A source-block in the sign-causality graph is a \emph{species block (S-block}) if each species node is adjacent to precisely two reaction nodes. A source-block in the sign-causality graph is a \emph{reaction block (R-block)} if each reaction node is adjacent to precisely two species nodes.
\end{definition}

	Proposition \ref{prop:SourceBlockSimplicity} tells us that when condition (ii) of Theorem \ref{thm:StrongConcordanceThm} is satisfied, every block within the sign-causality graph source is either an S-block or an R-block (or both in the case that the source-block is simply a single cycle). Note that in Figure \ref{fig:BlockTree} we have labeled the hypothetical source-blocks as $RBn$ or $SBn$ according to whether the corresponding source-block in Figure \ref{fig:BlockIllustration} is an R-block or an S-block.

	We conclude this section with two more propositions. Neither is essential to the proofs of the main theorems of this paper, but they provide some additional and not-so-obvious properties of a sign-causality graph source.

	The following proposition, proved in Appendix \ref{app:ProofOfSourceBlockProposition}, does not presuppose that condition (ii) of Theorem \ref{thm:ConcordanceTheorem} is satisfied. Rather, it tells us about properties of  R-blocks or S-blocks that might exist within a sign-causality-graph source. We already know that every directed cycle is even, but the proposition tells us that \emph{all} cycles within a source's R-blocks and S-blocks are even.

\begin{proposition}
\label{prop:EvenCyclesWithinSBlocksRBlocks} Every (not necessarily directed) cycle that lies within an S-block or an R-block of a sign-causality graph source is even. 
\end{proposition}

	The following proposition is a direct consequence of the two preceding ones:

\begin{proposition}
\label{prop:EvenCyclesWithinSource}
Suppose that, for the  reaction network under consideration, the Species-Reaction Graph satisfies condition (ii) of Theorem \ref{thm:ConcordanceTheorem}. Then in any source of the sign-causality graph every cycle is even.
\end{proposition}

\subsection{Properties of an end S-block}
\label{subsec:PropertiesOfEndSBlock}
 Recall that an \emph{end} S-block in a source is an S-block that contains at most one separating vertex of the source. We consider properties of such an end S-block, designated ESB. There are three mutually exclusive possibilities: ESB contains no separating vertex at all; it contains just one separating vertex, and it is a reaction vertex; or it contains just one separating vertex, and it is a species vertex. For our purposes the first two possibilities can be treated together, while the third requires other considerations. 

	In fact, we show that when condition (i) of Theorem \ref{thm:ConcordanceTheorem} holds, the first two possibilities cannot obtain; moreover, if the third obtains, we get sharpened information about the inequality in \eqref{eq:SourceInequalitySystem} corresponding to species at the block's separating vertex.

\subsubsection{Possibilities 1 and 2: ESB contains no separating vertex of the source or it contains a separating reaction vertex of the source} 
\label{subsec:ESBPossibilites1and2}
	Because it is strongly connected, ESB must contain a directed (and even) cycle, which we take to be 
\begin{equation}
\label{eq:SourceDirectedCycleESB}
s_1 \rightsquigarrow R_1  \rightsquigarrow s_2 \rightsquigarrow R_2 \dots\rightsquigarrow s_{n} \rightsquigarrow R_n \rightsquigarrow s_1.
\end{equation} 
Because ESB is a species-block and because the block contains no separating species-vertex of the source, each of the species in the block (and in the chosen cycle) is adjacent to at most two reactions of the block. Thus, the inequalities \eqref{eq:SourceInequalitySystem} corresponding to $s_1, s_2, \dots, s_n$ reduce to
\begin{eqnarray}
\label{eq:SBlockInequalitySequence}
f_{R_n \rightsquigarrow s_1}|\alpha_{R_n}| - e_{s_1 \rightsquigarrow R_1}|\alpha_{R_1}| &>& 0, \nonumber\\ 
f_{R_1 \rightsquigarrow s_2}|\alpha_{R_1}| - e_{s_2 \rightsquigarrow R_2}|\alpha_{R_2}| &>& 0,\\
&\vdots&\nonumber\\
f_{R_{n-1} \rightsquigarrow s_n}|\alpha_{R_{n-1}}| - e_{s_n \rightsquigarrow R_n}|\alpha_{R_n}| &>& 0. \nonumber
\end{eqnarray}
(Recall that, for a species-to-reaction edge $s \rightsquigarrow R$ of the sign-causality graph we denote by $e_{s \rightsquigarrow R}$ the stoichiometric coefficient of species $s$ in the corresponding edge-labeling complex. For a reaction-to-species edge $R \rightsquigarrow s$ we denote by $f_{R \rightsquigarrow s}$ the stoichiometric coefficient of species $s$ in its edge-labeling complex.) By sequentially invoking these inequalities from top to bottom, we can deduce from the last of them that
\begin{equation}
\label{eq:SBlockContradiction}
\left(\frac{ f_{R_1 \rightsquigarrow s_2}f_{R_2 \rightsquigarrow R_3} \dots f_{R_n \rightsquigarrow s_1}}{e_{s_1 \rightsquigarrow R_1}e_{s_2 \rightsquigarrow R_2} \dots e_{s_n \rightsquigarrow R_n}} - 1\right)\,|\alpha_{R_n}| > 0.
\end{equation}
However, when condition (i) of Theorem \ref{thm:ConcordanceTheorem} holds, \eqref{eq:SBlockContradiction} cannot obtain: By supposition $|\alpha_{R_n}|$ is positive. Given the $\rightsquigarrow$-orientation in the SR Graph, the even cycle \eqref{eq:SourceDirectedCycleESB} (viewed in the SR Graph) cannot be stoichiometrically expansive, so the first factor on the left of \eqref{eq:SBlockContradiction} is either zero or negative. (Recall \eqref{eq:NotStoichExpansiveCondition}.) Thus, we have a contradiction of \eqref{eq:SBlockContradiction}.

	We conclude, then, that when condition (i) of Theorem \ref{thm:ConcordanceTheorem} holds, an end S-block must contain a separating \emph{species} vertex of the source. We investigate next what happens in that case.

\subsubsection{Possibility 3: ESB contains a separating species vertex of the source} 
\label{subsec:ESBPossibility3}
Here again we consider a directed cycle in ESB, labeled as in  \eqref{eq:SourceDirectedCycleESB}. If  ESB's (unique) separating species vertex is not in the cycle, then we would again obtain a contradiction, just as in $\S$ \ref{subsec:ESBPossibilites1and2}. We suppose, then, that species $s_n$ is the separating vertex. Because ESB is a species-block, all other species vertices of the cycle are adjacent to precisely two reaction vertices, and those are in the cycle. Thus, all inequalities but the last in \eqref{eq:SBlockInequalitySequence} remain unchanged.

	On the other hand, $s_n$ is adjacent not only to $R_{n-1}$ and $R_n$ but also to certain reactions from nearby blocks sharing $s_n$ as a species. Recall that $\scrR_0$ is the set of all reactions in the source under study. We denote by $\scrR_{ESB}$ the set of reactions in ESB and by  $\scrR_0 \setminus \scrR_{ESB}$ the set of all source reactions not in ESB. Moreover, we let $\scrR_0 \setminus \scrR_{ESB} \rightsquigarrow s_n$ and  $s_n \rightsquigarrow \scrR_0 \setminus \scrR_{ESB}$ be the sets of source edges residing outside of ESB that are, respectively, directed toward and away from species $s_n$.

	In this case, the last inequality in \eqref{eq:SBlockInequalitySequence} must be replaced by
\begin{eqnarray}
f_{R_{n-1} \rightsquigarrow s_n}|\alpha_{R_{n-1}}| &-& e_{s_n \rightsquigarrow R_n}|\alpha_{R_n}|\;\; + \nonumber\\
\\ 
\sum_{\scrR_0 \setminus \scrR_{ESB} \rightsquigarrow s_n}f_{R \rightsquigarrow s_n}|\alpha_R|\; &-& \sum_{s_n \rightsquigarrow \scrR_0 \setminus \scrR_{ESB}}e_{s_n \rightsquigarrow R}|\alpha_R|\;\;>\; 0. \nonumber
\end{eqnarray}
Instead of \eqref{eq:SBlockContradiction}, this time we deduce the inequality
\begin{eqnarray}
\label{eq:SBlockAugmentedInequality}
\left(\frac{ f_{R_1 \rightsquigarrow s_2}f_{R_2 \rightsquigarrow R_3} \dots f_{R_n \rightsquigarrow s_1}}{e_{s_1 \rightsquigarrow R_1}e_{s_2 \rightsquigarrow R_2} \dots e_{s_n \rightsquigarrow R_n}} - 1\right)\,|\alpha_{R_n}| &+&\\
\sum_{\scrR_0 \setminus \scrR_{ESB} \rightsquigarrow s_n}f_{R \rightsquigarrow s_n}|\alpha_R|\; &-& \sum_{s_n \rightsquigarrow \scrR_0 \setminus \scrR_{ESB}}e_{s_n \rightsquigarrow R}|\alpha_R|\;>\; 0. \nonumber
\end{eqnarray}
When condition (i) of Theorem \ref{thm:ConcordanceTheorem} obtains, however, the first term on the left cannot be positive for reasons given in $\S$\ref{subsec:ESBPossibilites1and2}. Thus, we arrive at the following inequality, \emph{which relates entirely to source edges disjoint from ESB}:
\begin{equation}
\sum_{\scrR_0 \setminus \scrR_{ESB} \rightsquigarrow s_n}f_{R \rightsquigarrow s_n}|\alpha_R|\; - \sum_{s_n \rightsquigarrow \scrR_0 \setminus \scrR_{ESB}}e_{s_n \rightsquigarrow R}|\alpha_R|\;>\; 0. 
\end{equation}
For species $s_n$ this amounts to a sharpened form of its counterpart in  \eqref{eq:SourceInequalitySystem}, a form we will draw upon later on.
\subsection{Properties of an end R-block}
\label{subsec:PropertiesOfEndRBlock}
	We begin this subsection with an important proposition about R-blocks. A proof is provided in Appendix  \ref{app:RBlockAppendix}.

\begin{proposition}
\label{prop:RBlockMExistence}
Suppose that, in a sign-causality graph source, an R-block has species set $\scrS^*$, and suppose that no directed cycle in the block is stoichiometrically expansive relative to the $\rightsquigarrow$ orientation. Then there is a set of positive numbers $\{M_s\}_{s \; \in \scrS^*}$ such that, for each causal unit $s \rightsquigarrow R \rightsquigarrow s'$ in the block, the following relation is satisfied:
\begin{equation}
f_{R \rightsquigarrow s'}M_{s'} - e_{s \rightsquigarrow R}M_{s} \leq 0.
\end{equation} 
\end{proposition}

	Now we consider an end R-block in the source under consideration, designated ERB. We denote by \scrSst\; the set of species in ERB. In consideration of the source inequality system \eqref{eq:SourceInequalitySystem}, we restrict our attention to just those inequalities corresponding to species in \scrSst:
\begin{equation}
\label{eq:RBlockInequalitySystem}
\sum_{\scrR_0 \rightsquigarrow s}f_{R \rightsquigarrow s}|\alpha_R|\; - \sum_{s \rightsquigarrow \scrRo}e_{s \rightsquigarrow R}|\alpha_R| > 0, \quad \forall s \in \scrSst.
\end{equation}
We suppose that condition (i) of Theorem \ref{thm:ConcordanceTheorem} is satisfied, so, by virtue of Remark \ref{rem:StoichExpansiveRemark},  we can choose  $\{M_s\}_{s \; \in \scrS^*}$ as in Proposition \ref{prop:RBlockMExistence}. If we multiply each inequality in \eqref{eq:RBlockInequalitySystem} by the corresponding $M_s$ and sum, we get  the single inequality shown in \eqref{eq:RBlockSingleInequality}.
\begin{equation}
\label{eq:RBlockSingleInequality}
\sum_{s \in \scrSst}\left(\sum_{\scrR_0 \rightsquigarrow s}f_{R \rightsquigarrow s}M_s|\alpha_R|\; - \sum_{s \rightsquigarrow \scrRo}M_s e_{s \rightsquigarrow R}|\alpha_R|\right) > 0
\end{equation}

As in $\S$\ref{subsec:PropertiesOfEndSBlock} there are three possibilities: ERB contains no separating vertex; it contains just one separating vertex, and it is a reaction vertex; or it contains just one separating vertex, and it is a species vertex. We will show that neither of the first two possibilities can obtain. Then, as in $\S$\ref{subsec:ESBPossibility3}, we will show that, if the third possibility is realized, the inequality in \eqref{eq:SourceInequalitySystem} corresponding to the species at the separating vertex can be sharpened to considerable advantage.

\subsubsection{Possibilities 1 and 2: ERB contains no separating vertex of the source or it contains a separating reaction vertex of the source} 
\label{subsec:ERBPossibilites1and2} 
Because ERB is an R-block, each reaction is adjacent to precisely two species in the set \scrSst, which is to say that each reaction in ERB sits at the center of precisely one causal unit in ERB. (Recall that ERB is strongly connected.) Moreover, if either of the first two  possibilities should obtain, no species is adjacent to a reaction not in ERB. In these cases, the inequality \eqref{eq:RBlockSingleInequality} can be rewritten. Let \scrU\;be the set of causal units within ERB. Then \eqref{eq:RBlockSingleInequality} can be made to take the form shown in \eqref{eq:RBlockSingleInequality2}.
\begin{equation}
\label{eq:RBlockSingleInequality2}
\sum_{s \rightsquigarrow R \rightsquigarrow s' \in \scrU}(f_{R \rightsquigarrow s'}M_{s'} - e_{s \rightsquigarrow R}M_{s})|\alpha_R|\; >\; 0
\end{equation}
However, \eqref{eq:RBlockSingleInequality2} is contradicted by the attributes given to the set  $\{M_s\}_{s \; \in \scrS^*}$  in Proposition \ref{prop:RBlockMExistence}. 

	Thus, neither of the first two possibilities can obtain. ERB must contain a separating species vertex of the source. We examine next what can be said in that case.

\subsubsection{Possibility 3: ERB contains a separating species vertex of the source} Suppose that ERB contains a species vertex $s^*$ that is a separating vertex of the source under consideration. When condition (i) of Theorem \ref{thm:ConcordanceTheorem} obtains, we can again choose  positive numbers  $\{M_s\}_{s \; \in \scrS^*}$ to satisfy the requirements of Proposition \ref{prop:RBlockMExistence}, and the inequality \eqref{eq:RBlockSingleInequality} remains in force. On the other hand the passage from \eqref{eq:RBlockSingleInequality} to \eqref{eq:RBlockSingleInequality2} becomes confounded by the fact that species vertex $s^*$ is now adjacent to source edges not residing in ERB. We denote by $\scrR_{ERB}$ the set of reactions in ERB and by  $\scrR_0 \setminus \scrR_{ERB}$ the set of all source reactions not in ERB. Moreover, we let $\scrR_0 \setminus \scrR_{ERB} \rightsquigarrow s^*$ and  $s^* \rightsquigarrow \scrR_0 \setminus \scrR_{ERB}$ be the sets of source edges residing outside of ERB that are, respectively, directed toward and away from species $s^*$. As before, we let \scrU\  be the set of causal units in ERB. In this case, \eqref{eq:RBlockSingleInequality} can be recast as \eqref{eq:RBlockSingleInequality3}.
\begin{eqnarray}
\label{eq:RBlockSingleInequality3}
\sum_{s \rightsquigarrow R \rightsquigarrow s' \in \scrU}(f_{R \rightsquigarrow s'}M_{s'} &-& e_{s \rightsquigarrow R}M_{s})|\alpha_R| \;\; + \nonumber\\ \smallskip\\
 M_{s^*}(\sum_{\scrR_0 \setminus \scrR_{ERB} \rightsquigarrow s^*}f_{R \rightsquigarrow s^*}|\alpha_R| &-& \sum_{s^* \rightsquigarrow \scrRo \setminus \scrR_{ERB}} e_{s^* \rightsquigarrow R}|\alpha_R|)  \; >\; 0 \nonumber
\end{eqnarray}  

	Recall that the set   $\{M_s\}_{s \; \in \scrS^*}$ was chosen to satisfy the requirements of Proposition \ref{prop:RBlockMExistence}, so the first sum in \eqref{eq:RBlockSingleInequality3} cannot be positive. Then, because $M_{s^*}$ is positive, we must have
\begin{equation}
\sum_{\scrR_0 \setminus \scrR_{ERB} \rightsquigarrow s^*}f_{R \rightsquigarrow s^*}|\alpha_R| \;\;- \sum_{s^* \rightsquigarrow \scrRo \setminus \scrR_{ERB}} e_{s^* \rightsquigarrow R}|\alpha_R|  \; >\; 0.
\end{equation}
Note that this is a strengthened counterpart of the inequality in \eqref{eq:SourceInequalitySystem} corresponding to species $s^*$, \emph{a counterpart that makes reference only to reactions external to the end reaction block ERB}.

\subsection{The concluding argument: Leaf removal} We begin this subsection with a review of what was established in $\S$\ref{subsec:PropertiesOfEndSBlock} and $\S$\ref{subsec:PropertiesOfEndRBlock}: An end block in a sign-causality graph source, be it an S-block or an R-block,  must contain a \emph{species} separating vertex of the source. This implies that the hypothetical source depicted in Figure \ref{fig:BlockIllustration} (and having block-tree depicted in Figure \ref{fig:BlockTree}) cannot in fact be a source, for it has an end block, corresponding to RB1 in Figure \ref{fig:BlockTree}, that does not contain a separating \emph{species} vertex of the putative source. (Each of the remaining end blocks does contain a separating species vertex.)

	Moreover, we have established that the inequality in \eqref{eq:SourceInequalitySystem} corresponding to a separating species vertex $s^*$ in an end block EB, be it an S-block or an R-block, can be strengthened to a form that makes no mention of reaction inside EB:
\begin{equation}
\label{eq:StrengthenedSourceInequality}
\sum_{\scrR_0 \setminus \scrR_{EB} \rightsquigarrow s^*}f_{R \rightsquigarrow s^*}|\alpha_R| \;\;- \sum_{s^* \rightsquigarrow \scrRo \setminus \scrR_{EB}} e_{s^* \rightsquigarrow R}|\alpha_R|  \; >\; 0.
\end{equation}
Here $\scrR_{EB}$ is the set of reactions in EB, and   $\scrR_0 \setminus \scrR_{EB}$ the set of source reactions not in EB. Moreover,  $\scrR_0 \setminus \scrR_{EB} \rightsquigarrow s^*$ and  $s^* \rightsquigarrow \scrR_0 \setminus \scrR_{EB}$ are the sets of source edges residing outside of EB that are, respectively, directed toward and away from species $s^*$.

	Now if EB is an end block in the source under consideration, we can replace the inequality in \eqref{eq:SourceInequalitySystem} corresponding to the unique separating species vertex $s^*$ in EB with its strengthened form shown in \eqref{eq:StrengthenedSourceInequality}. Thereafter, we can restrict the now-modified inequality system \eqref{eq:SourceInequalitySystem} to just those inequalities corresponding to  $s^*$ and to species residing outside EB. In effect, the resulting reduced system of inequalities corresponds to a subgraph of the source with the end block EB removed, but with $s^*$ retained. Viewed in the source's block tree, this corresponds to removal of a leaf along with that leaf's adjacent species separating vertex (when that separating vertex is not adjacent to a different leaf).

	It is not difficult to see that the arguments in $\S$\ref{subsec:PropertiesOfEndSBlock} or $\S$\ref{subsec:PropertiesOfEndRBlock} can then be applied to any end block EB$^\prime$ of the resulting source subgraph to produce a still smaller but strengthened inequality system corresponding to a still smaller source subgraph, a subgraph resulting from removal of EB$^\prime$. 

	The process can be continued, amounting to a sequential pruning of the source's block-tree, with each stage corresponding to removal of a (perhaps new) leaf and perhaps its adjacent species separating vertex. The process will terminate when only one block remains, with the final source subgraph having no separating vertex at all.  In this case, by arguments of $\S$\ref{subsec:ESBPossibilites1and2} or  $\S$\ref{subsec:ERBPossibilites1and2} the corresponding reduced inequality system cannot be satisfied, and we have a contradiction.

	This completes the proof of Theorem \ref{thm:ConcordanceTheorem}.

\section{Proof of Theorem \ref{thm:StrongConcordanceThm}} 
\label{sec:StrongConcordance}
Here we prove Theorem \ref{thm:StrongConcordanceThm}, which is repeated below:
\bigskip

\noindent
\textbf{Theorem \ref{thm:StrongConcordanceThm}} \emph{A fully open reaction network is strongly concordant if its true-chemistry Species-Reaction Graph has the following properties:}
\medskip

\emph{ (i) Every even cycle is an s-cycle.}
\medskip

\emph{(ii) No two even cycles have a species-to-reaction intersection.}
\bigskip

	We begin by recalling the definition of strong concordance for an arbitrary reaction network \rnet, not necessarily fully open, with stoichiometric subspace $S \subset \RS$. Again we let $L: \RR \to S$ be the linear map defined by
\begin{equation}
L\alpha = \sum_{\rxn \in \scrR}\alpha_{\rxn}(y' - y).
\end{equation}

\begin{definition}
\label{DEF:StronglyConcordant}
A reaction network \rnet with stoichiometric subspace $S$ is \emph{strongly concordant} if there do not exist $\alpha \in \ker L$ and a non-zero $\sigma \in S$ having the following properties:

\begin{enumerate}[(i)]
\item{For each $y \to y'$ such that $\alpha_{y \to y'} > 0$ there exists a species $s \in \supp (y - y')$ for which $\sgn \sigma_s = \sgn (y - y')_s.$}

\item{For each $y \to y'$ such that $\alpha_{y \to y'} < 0$ there exists a species $s \in \supp (y - y')$ for which $\sgn \sigma_s = - \sgn (y - y')_s$.}

\item{For each $y \to y'$ such that $\alpha_{y \to y'} = 0$, either (a) $\sigma_s = 0$ for all $s \in \supp y$, or (b) there exist species $s, s' \in \supp (y - y')$ for which $\sgn \sigma_s = \sgn (y - y')_s$ and $\sgn \sigma_{s'} = -  \sgn (y - y')_{s'}$.}
\end{enumerate}
\end{definition}

	The following lemma will take us just short of a proof of Theorem \ref{thm:StrongConcordanceThm}. Proof of the lemma is given in Appendix \ref{app:ProofofStrongConcordanceLemme}.

\begin{lemma}
\label{lem:StrongConcordanceLemma}
Suppose that a fully open reaction network with true chemistry \newline  \rnet is not strongly concordant. Then there is another true chemistry \rnetbar whose fully open extension is discordant and whose SR Graph is identical to a subgraph of the SR Graph for \rnet, apart perhaps from changes in certain arrow directions within the reaction vertices.
\end{lemma}

	Proof of Theorem \ref{thm:StrongConcordanceThm} proceeds from Lemma \ref{lem:StrongConcordanceLemma} in the following way:

\begin{proof}[Proof of Theorem \ref{thm:StrongConcordanceThm}] Suppose that the SR Graph for a true chemistry \newline \rnet\  satisfies conditions (i) and (ii) of the theorem but that, contrary to what is to be proved, the fully open extension of \rnet\ is not strongly concordant. Note that when the SR Graph of a true chemistry \rnet\; satisfies conditions (i) and (ii) of Theorem \ref{thm:StrongConcordanceThm}, so will any subgraph of that SR Graph. On the other hand, neither of those conditions depends upon the direction of arrows in the reaction nodes. Thus, the SR Graph of the true chemistry \rnetbar\  given by the lemma will also satisfy conditions (i) and (ii). But then, by Corollary \ref{cor:ConcordanceThmCorollary} of Theorem \ref{thm:ConcordanceTheorem}, the fully open extension of \rnetbar\ could not be discordant, and we have a contradiction.
\end{proof}

\section{Extensions of the main theorems to networks that are not fully open}
\label{sec:Extensions}	

	It is the purpose of this Section to elaborate on remarks made in $\S$\ref{subsec:NotFullyOpen}. 

	When the SR Graph drawn for a true chemical reaction network satisfies the  hypotheses of Theorems \ref{thm:ConcordanceTheorem} or \ref{thm:StrongConcordanceThm}, these theorems tell us that the network's fully open extension is concordant (or strongly concordant). In this case, the fully open network inherits the many attributes described in \cite{shinar_concordant_2011} that accrue to all concordant (or strongly concordant) networks. We would like to know circumstances under which these theorems can be extended in their range to give concordance information about networks that are not fully open. 

	More generally, we would like to know conditions under which, for a given network, concordance or strong concordance of the network's fully open extension implies concordance or strong concordance of the network itself.\footnote{We do not preclude the possibility that the original network contain \emph{some} degradation reactions of the form $s \to 0$.} This is a question quite separate from SR Graph considerations. However, when the network satisfies such conditions and its underlying true chemistry SR Graph satisfies the hypothesis of Theorems  \ref{thm:ConcordanceTheorem} or \ref{thm:StrongConcordanceThm}, then the concordance properties ensured by those theorems for the fully open extension will be inherited by the original network.

	In \cite{shinar_concordant_2011} we showed that a \emph{normal} network is concordant (strongly concordant) if its fully open extension is concordant (strongly concordant). Normality is a mild structural condition given in  Definition \ref{def:normal} below. In \cite{craciun_multiple_2010} it was shown that \emph{every weakly reversible network is normal}. 

	Theorems \ref{thm:SmallToLargeConcordanceThm} and \ref{thm:SmallToLargeStronglyConcordanceThm} below tell us that these same results also obtain for the still larger class of \emph{weakly normal} networks. (See Definition \ref{def:weaklynormal}.)  In particular, a weakly normal network is concordant (strongly concordant) if its fully open extension is concordant (strongly concordant). 

	This improvement on results in \cite{shinar_concordant_2011} is helpful in itself, but it also has significance in another direction: We show in $\S$\ref{subsec:normality-weaknormality} that the class of weakly normal networks is synonymous with the very broad class of \emph{nondegenerate} networks (Definition \ref{def:nondegenerate}), which was described in the Introduction. Thus, \emph{any nondegenerate reaction network with a concordant (strongly concordant) fully open extension is itself concordant (strongly concordant)}. Moreover, we also show that, with respect to the possibility of concordance, \emph{degenerate} reaction networks are not worth considering, for they are never concordant.

	In $\S$\ref{subsec:computationaltests} we provide computational tests that serve to affirm network normality and weak normality (or, equivalently nondegeneracy).

\subsection{Network normality, weak normality, and nondegeneracy}
\label{subsec:normality-weaknormality}
\begin{definition}
\label{def:normal}
Consider a reaction network \rnet \  with stoichiometric subspace $S$. The network is \emph{normal} if there  are $q \in \PS$ and $\eta \in \PR$ such that the linear transformation $T: S \to S$ defined by
\begin{equation}
\label{eq:normalityT}
T \sigma := \sum_{y \to y' \in \scrR}\eta_{y \to y'}(y * \sigma)(y' - y)
\end{equation}
is nonsingular, where ``$*$" is the scalar product in \RS \  defined by
\begin{equation}
x*x' := \sum_{s \in \scrS}q_sx_sx'_s. 
\end{equation}
\end{definition}

\begin{rem} As indicated earlier, it was shown in \cite{craciun_multiple_2010} that every weakly reversible network is normal. Reference \cite{craciun_multiple_2010} also contains structural conditions that ensure normality for certain ``partially open" networks that are not weakly reversible.
\end{rem}
\bigskip

	In preparation for the next definition we note that \eqref{eq:normalityT} can be written as
\begin{equation}
\label{eq:normalityT2}
T \sigma := \sum_{y \to y' \in \scrR}(\eta_{y \to y'}y \circ q) \cdot \sigma\ (y' - y),
\end{equation}
where ``$\cdot$" indicates the standard scalar product in \RS.
\medskip
\begin{definition}
\label{def:weaklynormal}
Consider a reaction network \rnet \  with stoichiometric subspace $S$. The network is \emph{weakly normal} if, for each reaction \rxn, there  is a vector $p_{\rxn} \in \PbarS$ with $\supp  p_{\rxn} = \supp y$ such that the linear transformation $\bar{T}: S \to S$ defined by
\begin{equation}
\label{eq:WeaklyNormalTbar}
\bar{T} \sigma := \sum_{y \to y' \in \scrR}p_{y \to y'} \cdot \sigma(y' - y)
\end{equation}
is nonsingular. Here  ``$\cdot$" is the standard scalar product in \RS.
\end{definition}

\begin{rem} A reaction network that is normal is also weakly normal. In fact, if  
$q \in \PS$ and $\eta \in \PR$ satisfy the requirements of Definition \ref{def:normal}, then the choice $p_{\rxn} := \eta_{\rxn}y\circ q, \forall \rxn \in \scrR$ will satisfy the requirements of Definition \ref{def:weaklynormal}. On the other hand, a weakly normal network need not be normal. An example is given by
\begin{equation}
\label{eq:tminuslcounterexample}
C \leftarrow A + B \rightarrow D \rightarrow 2A,
\end{equation}
which is weakly normal but not normal.  Because every weakly reversible network is normal, it follows that every weakly reversible network is also weakly normal.
\end{rem}
\medskip

\begin{rem} For readers familiar with standard language of chemical reaction network theory \cite{feinberg_lectureschemical_1979}, a network cannot be normal if 
\begin{equation}
\label{eq:t-lcondition}
t - \ell - \delta \ > \ 0, 
\end{equation}
where $t$ is the number of terminal strong linkage classes, $\ell$ is the number of linkage classes, and $\delta$ is the deficiency. This follows without much difficulty from \cite{feinberg_chemical_1977}; see also \cite{feinberg_lectureschemical_1979}. For network \eqref{eq:tminuslcounterexample}, $t = 2,\  \ell = 1$, and $\delta = 0$, so normality is precluded by condition \eqref{eq:t-lcondition}. Network \eqref{eq:tminuslcounterexample} illustrates, however, that the same condition does not also preclude weak normality. 
\end{rem}

\begin{definition}
\label{def:nondegenerate}
A reaction network is \emph{nondegenerate} if there exists for it differentiably monotonic kinetics ($\S$\ref{subsec:Definitions}) such that at some positive composition $c^*$ the derivative of the species-formation-rate function $df(c^*):S \to S$ is nonsingular. Otherwise, the network is \emph{degenerate}.
\end{definition}

\begin{proposition}
\label{prop:WeaklyNormalEqualsNondegenerate}
A reaction network is nondegenerate if and only if it is weakly normal.
\end{proposition}

\begin{proof} Suppose that a network \rnet \  is nondegenerate. Then there is for the network a kinetics \scrK \ such that, at some composition $c^* \in \PS$, the kinetics is differentiably monotonic and, moreover, the derivative of the species-formation-rate function,  $df(c^*):S \to S$, is nonsingular. In this case, for each $\sigma \in S$ 
\begin{equation}
\label{eq:deriv}
df(c^*)\sigma = \sum_{\rxn \in \scrR}\nabla\scrK_{\rxn}(c^*)\cdot\sigma(y' - y),
\end{equation}
where the components of $\nabla\scrK_{\rxn}(c^*)$ have the non-negativity properties that follow from differentiable monotonicity ($\S$\ref{subsec:Definitions}). By taking 
\[p_{\rxn} = \nabla\scrK_{\rxn}(c^*), \quad \forall \rxn \in \scrR\]
we can establish that the network is weakly normal. 

	On the other hand, suppose that the network \rnet is weakly normal and, in particular, that the set $\{p_{\rxn}\}_{\rxn \in \scrR}$ satisfies the requirements of Definition \ref{def:weaklynormal}.  Let \scrK\  be the (differentiably monotonic) kinetics defined by
\[\scrK_{\rxn}(c) := c^{p_{\rxn}}, \quad \forall \rxn \in \scrR, \] 
and let $c^* \in \PS$ be such that $c^*_s = 1$ for each $s \in \scrS.$ Note that 
\[\nabla\scrK_{\rxn}(c^*) = p_{\rxn}, \quad \forall \rxn \in \scrR.\]
\noindent 
From this, \eqref{eq:deriv}, and the properties of the set $\{p_{\rxn}\}_{\rxn \in \scrR}$ given by Definition \ref{def:weaklynormal} it follows that $df(c^*):S \to S$, is nonsingular, whereupon the network is nondegenerate.                               
\end{proof}

\begin{rem}
Note that in the proof that nondegeneracy implies weak normality we did not actually require that the network be nondegenerate. In particular, we did not require that the kinetics \scrK \  be differentiably monotonic at \emph{all} positive compositions, only that it be differentiably monotonic at \emph{one} composition, $c^*$ (and, of course, that $df(c^*)$ be nonsingular). However, when  these apparently milder conditions are invoked, the network \emph{must} be nondegenerate nevertheless: The seemingly milder conditions result in weak normality, and, as the second part of the proof indicates, weak normality implies nondegeneracy.
\end{rem}
\medskip

	The following proposition indicates that a network that is not weakly normal  (or, equivalently, is degenerate) has no chance of being concordant.

\begin{proposition}
\label{prop:DegeneracyImpliesDiscordance}
A reaction network that is not weakly normal is discordant. Equivalently, every degenerate network is discordant.
\end{proposition}

\begin{proof}
	Suppose that a reaction network \rnet \ is not weakly normal (and, in particular, is not normal). From Definition \ref{def:weaklynormal} it follows that, for the special choice $p_{\rxn} = y, \; \forall \rxn \in \scrR$,\  the corresponding map  $\bar{T}: S \to S$ given by \eqref{eq:WeaklyNormalTbar} must be singular. This is to say that there is a nonzero $\sigma^* \in S$ such that
\begin{equation}
 \sum_{y \to y' \in \scrR}y \cdot \sigma^*(y' - y) = 0. \nonumber
\end{equation} 
Now let $\alpha \in \RR$ be defined by $\alpha_{\rxn} = y\cdot \sigma^*, \; \forall \rxn \in \scrR$. Then, in view of Definition \ref{def:Concordant}, the pair consisting of $\alpha$ and $\sigma^*$ serve to establish discordance of the network under consideration. 
\end{proof}

\subsection{Concordance of a network and of its fully open extension}
\label{subsec:ThmsOnSmallToLarge}
	The following theorems about weakly normal networks amount to straightforward extensions of theorems in \cite{shinar_concordant_2011} about normal networks; the proofs are almost identical. Although these theorems give the former ones a somewhat greater range, their main interest lies in the fact that they can be stated in terms of the more tangible, but equivalent, notion of network nondegeneracy.

\begin{theorem}
\label{thm:SmallToLargeConcordanceThm}
A weakly normal (or, equivalently, nondegenerate) network is concordant if its fully open extension is concordant. In particular, a weakly reversible network is concordant if its fully open extension is concordant.
\end{theorem}
\medskip
\begin{theorem}
\label{thm:SmallToLargeStronglyConcordanceThm}
A weakly normal (or, equivalently, nondegenerate) network is  strongly concordant if its fully open extension is strongly concordant. In particular, a weakly reversible network is strongly concordant if its fully open extension is strongly concordant.
\end{theorem}

	Taken together, Proposition \ref{prop:WeaklyNormalEqualsNondegenerate}, Proposition \ref{prop:DegeneracyImpliesDiscordance}, Theorem \ref{thm:SmallToLargeConcordanceThm}, and Theorem \ref{thm:SmallToLargeStronglyConcordanceThm} tell us that, in the class of networks that have a concordant (strongly concordant) fully open extension, the concordant (strongly concordant) ones are \emph{precisely} the nondegenerate ones:

\begin{corollary}
Consider a reaction network that has a concordant (strongly concordant) fully open extension. Then the original network is concordant (strongly concordant) if and only if it is nondegenerate.
\end{corollary}

\begin{rem}
It is a consequence of Theorem \ref{thm:SmallToLargeConcordanceThm} and results in \cite{shinar_concordant_2011} that the dynamical statements (i) - (iii) of Theorem \ref{thm:DynamicalTheorem} hold true for any nondegenerate reaction network whose fully open extension is concordant, not merely those that satisfy the SR Graph conditions of  Theorem \ref{thm:ConcordanceTheorem}. In particular, the ``all or nothing" observation of Remark \ref{rem:AllOrNothing} still obtains. The SR Graph conditions of Theorem \ref{thm:ConcordanceTheorem} merely suffice to ensure concordance of the fully open extension. 
\end{rem}

\subsection{Tests for network normality, weak normality, and nondegeneracy}
\label{subsec:computationaltests}

	Here we provide some computational means to affirm normality and weak normality (or, equivalently, nondegeneracy) of a reaction network. Recall that the rank of a network is the rank of its set of reaction vectors.

\begin{proposition}
\label{prop:PMatrixTest}
A reaction network \rnet of rank $r$ is weakly normal (or, equivalently, nondegenerate) if there is a set of $r$ reactions  \mbox{$\{y_i \to y'_i\}_{i=1 \dots r}$} and a set of vectors  $\{p_i\}_{i=1 \dots r} \subset \PbarS$ with $\supp p_i = \supp y_i$, $i = 1 \dots r$,  such that the matrix 
\begin{equation}
\label{eq:p_iMatrix}
\left[ p_i\cdot (y'_j - y_j)\right]_{i,j = 1 \dots r}
\end{equation}
has nonzero determinant. 
\end{proposition}

	Proof of the proposition is provided at the end of this subsection. For a particular choice of $r$ reactions one can readily construct the matrix \eqref{eq:p_iMatrix} in terms of symbols for the species-wise components of the vectors  $\{p_i\}_{i=1 \dots r}$, and one can then calculate the determinant of the matrix as a polynomial in those same symbols.  If, for \emph{even one} choice of $r$ reactions, the resulting determinant   is not \emph{identically} zero, then the network is weakly normal. 

	We remark in passing that a nonzero determinant will require that the $r$ reaction vectors \mbox{$\{y'_i - y_i\}_{i=1 \dots r}$}  be linearly independent. Of course, such independent reaction vector sets will invariably exist for a network of rank $r$.

	A special choice of $\{p_i\}_{i=1 \dots r}$, one that obviates the need for symbolic computation, is invoked in the following corollary. This choice will often suffice to establish weak normality. In fact, when the condition in Corollary \ref{cor:NormalityCondition} below is satisfied, the network will not only be weakly normal but also normal. (See Remark \ref{rem:NormalityRemark} following the proof of Proposition \ref{prop:PMatrixTest}.)  

\begin{corollary}
\label{cor:NormalityCondition}
A reaction network \rnet of rank $r$ is weakly normal (and, in fact, normal) if there is a set of $r$ reactions \mbox{$\{y_i \to y'_i\}_{i=1 \dots r}$} such that the matrix 
\begin{equation}
\label{eq:y_iMatrix}
\left[y_i\cdot (y'_j - y_j)\right]_{i,j = 1 \dots r}
\end{equation}
has nonzero determinant.
\end{corollary}

\begin{example}
Here we apply Corollary \ref{cor:NormalityCondition} to network \eqref{eq:NiceNetworkExample}, the rank of which is 4. For the four-reaction set 
\begin{equation}
\{P \to A + B, \; Q \to B + C,\; 2A \to C,\; C + D \to Q + E\} ,
\end{equation}
it is easy to calculate that the matrix in the corollary has a determinant of 4, so the network is weakly normal. Recall that Theorem \ref{thm:ConcordanceTheorem} established the concordance of the fully open extension of network \eqref{eq:NiceNetworkExample}. Because network \eqref{eq:NiceNetworkExample} is weakly normal, concordance of its fully open extension extends to network \eqref{eq:NiceNetworkExample} itself. 
\end{example}

\begin{rem} Corollary \ref{cor:NormalityCondition} is only one consequence of Proposition \ref{prop:PMatrixTest}. There are other, more interesting ones (some with a graphical flavor) that we intend to take up in another article. For example, a network \rnet\  of rank $r$ is weakly normal if there exist $r$ reactions  \mbox{$\{y_i \to y'_i\}_{i=1 \dots r}$} with the following property: There is a set of distinct species $\scrS_{*} = \{s_1,\dots,s_r\}$ with $s_i \in \supp y_i, \; i = 1 \dots r,$ such that
\begin{equation}
\label{eq:StrippedSpeciesEqn}
\det \left[s_i\cdot (y'_j - y_j)\right]_{i,j = 1 \dots r} \neq 0.
\end{equation}
It is not difficult to see that \eqref{eq:StrippedSpeciesEqn} has the following interpretation: Let  $\{\bar{y}_i \to \bar{y}'_i\}_{i=1 \dots r}$ denote the set of ``reactions" obtained from  $\{y_i \to y'_i\}_{i=1 \dots r}$ by stripping away all species not in $\scrS_{*}$.  Then  \eqref{eq:StrippedSpeciesEqn} is satisfied (whereupon the original network is weakly normal) precisely when the resulting set of ``reaction vectors"  $\{\bar{y}'_i - \bar{y}_i\}_{i=1 \dots r}$ is linearly independent.

\end{rem}

\begin{proof}[Proof of Proposition \ref{prop:PMatrixTest}] Suppose that the reaction vector set \mbox{$\{y'_i - y_i\}_{i=1 \dots r}$} and the set $\{p_i\}_{i=1 \dots r} \subset \PbarS$ satisfy the conditions of the Proposition \ref{prop:PMatrixTest}. In this case, it is not difficult to see that each of these sets must be linearly independent. In particular, the set  \mbox{$\{y'_i - y_i\}_{i=1 \dots r}$} is a basis for $S$, the stoichiometric subspace for the network under consideration.

	We begin by constructing the linear transformation $\bar{T}_0:S \to S$ defined by
\begin{equation}
\bar{T}_0\,\sigma = \sum_{i=1}^r\ p_i \cdot \sigma\,(y'_i - y_i), \quad \forall \sigma \in S,
\end{equation}
which can be seen to be nonsingular in the following way: Suppose on the contrary that there is a nonzero $\sigma^* \in S$ such that $\bar{T}_0\,\sigma^* = 0$. Because the set  \mbox{$\{y'_i - y_i\}_{i=1 \dots r}$} is linearly independent, we must have 
\begin{equation}
\label{eq:pi_dot_sigma}
p_i \cdot \sigma^* = 0, \quad i = 1 \dots r.
\end{equation}
Because \mbox{$\{y'_j - y_j\}_{j=1 \dots r}$}  is a basis for $S$ and $\sigma^*$ is a nonzero member of $S$, there must be $\xi_j, \ j = 1 \dots r$, not all zero, such that
\begin{equation}
\sigma^* = \sum_{j=1}^r\ \xi_j\,(y'_j - y_j).
\end{equation}
Insertion of this into \eqref{eq:pi_dot_sigma} results in the system of $r$ homogeneous equations
\begin{equation}
\sum_{j=1}^r\ p_i\cdot (y'_j - y_j)\,\xi_j  = 0, \quad  \quad i = 1 \dots r
\end{equation}
that must be satisfied by the set $\{\xi_j\}_{j=1\dots r}$. Since the determinant of the matrix \eqref{eq:p_iMatrix} is nonzero, the only solution is 
$\xi_j = 0, \; j=1\dots r$, which amounts to a contradiction. Thus, $\bar{T}_0$ is nonsingular, and, as a result, $\det \bar{T}_0 \neq 0.$

	It remains to be shown that the requirements of weak normality are met by the network \rnet. For this purpose, let $\scrR_0$ be the set of aforementioned reactions \mbox{$\{y_i \to y'_i\}_{i=1 \dots r}$}. Moreover, let $p_{\rxn} \in \PbarS$  be chosen to satisfy
\begin{eqnarray}
p_{y_i \to y'_i} := p_i, \; &&\forall y_i \to y'_i \in \scrR_0 \nonumber\\
p_{y \to y'} := \epsilon y, \; &&\forall \rxn \in \scrR \setminus \scrR_0,\nonumber
\end{eqnarray}
where $\epsilon$ is a small positive number.  Now let  $\bar{T}_\epsilon:S \to S$ be defined by 
\begin{eqnarray}
\bar{T}_\epsilon\, \sigma &:=& \sum_{y \to y' \in \scrR}p_{y \to y'} \cdot \sigma\,(y' - y)\nonumber\\
&=& \bar{T}_0\,\sigma + \;\;\epsilon\left[\sum_{y \to y' \in \scrR \setminus \scrR_0} y \cdot \sigma\,(y' - y)\right]\nonumber.
\end{eqnarray}
Note that $\bar{T}_{\epsilon}|_{\epsilon = 0} = \bar{T}_0$. Because $\det \bar{T}_0 \neq 0$ and because $\det \bar{T}_{\epsilon}$ is continuous in $\epsilon$, it follows that  $\det \bar{T}_{\epsilon} \neq 0$ for $\epsilon$ sufficiently small. Thus, for sufficiently small  $\epsilon$, $\bar{T}_{\epsilon}$ is nonsingular,whereupon the network \rnet is weakly normal.
\end{proof}

\begin{rem}
\label{rem:NormalityRemark}
Corollary \ref{cor:NormalityCondition} derives from Proposition \ref{prop:PMatrixTest} by invoking the special  choice $p_i = y_i,\ i = 1 \dots r$. When the resulting condition in Corollary \ref{cor:NormalityCondition} is satisfied, the network at hand is not only weakly normal but also normal. To see this, it is enough to replace $p_i$ by $y_i$ everywhere in the proof of Proposition \ref{prop:PMatrixTest} and then invoke Definition \ref{def:normal} with $q_s = 1, \; \forall s \in \scrS$.
\end{rem}

\section{A Concluding Remark}
	When their hypotheses are satisfied, the central theorems of this article permit one to affirm concordance of a particular reaction network from inspection of its Species Reaction Graph and, then, to invoke all of the powerful dynamical consequences that concordance implies \cite{shinar_concordant_2011}. At the same time, these theorems provide  delicately nuanced insight into network attributes that give rise to concordance. And, because their hypotheses are fairly easy to satisfy, the theorems  tell us  more --- that concordance in realistic chemical reaction networks is likely to be common.  Moreover, the theorems also serve to make connections between concordance results in \cite{shinar_concordant_2011} and earlier, somewhat different SR-Graph-related results contained in \cite{schlosser_theory_1994,craciun_multiple_2006-1,craciun_understanding_2006-1,banaji_graph-theoretic_2009,banaji_graph-theoretic_2010}.

	It should be remembered, however, that computational means for direct  determination of whether a particular network  --- \emph{fully open or otherwise} --- is concordant \emph{or discordant}, strongly concordant \emph{or not strongly concordant}, are already available in user-friendly, freely-provided software \cite{Ji_toolboxWinV21} prepared in connection with \cite{shinar_concordant_2011}.  In most instances, this software or a variant of it, will be the tool of choice.

\pagebreak
\addcontentsline{toc}{section}{Appendices}
\appendix
\numberwithin{equation}{section}
\appendixpage
\section{Proof of Propositions \ref{prop:SourceBlockSimplicity} and \ref{prop:EvenCyclesWithinSBlocksRBlocks}}
\label{app:ProofOfSourceBlockProposition}

\subsection{Some graph theoretical preliminaries: Ears}
\label{sec:Ears}
In this sub-section we describe a small amount of graph-theoretical material, all of which is available in \cite{bondy_graph_2010}. In a directed graph $G$ let $F$ be a proper subgraph of $G$ (with edge-directions inherited from $G$). A  \emph{directed ear} of $F$ in $G$ is a directed path in $G$ whose end vertices lie in $F$ but whose other vertices do not. In Figure \ref{fig:EarIllustration} we show a hypothetical source block, and we consider the proper subgraph with vertices $S_1,S_2,R_1,R_2,R_3$ and the directed edges connecting them. Within the full graph, the path $S_1 \leadsto R_4 \leadsto S_3 \leadsto R_2$ is a directed ear of that subgraph.

\begin{figure}[ht]
\centering
\includegraphics[scale=.6]{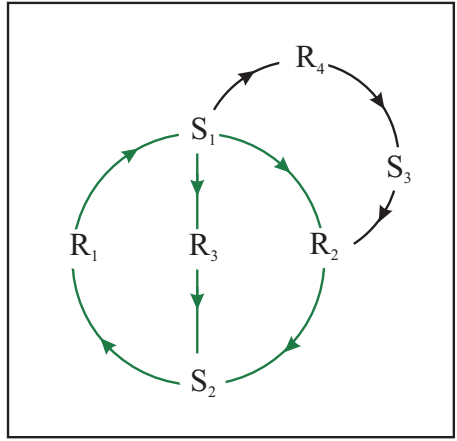}
\caption{A directed ear}
\label{fig:EarIllustration}
\end{figure}

	Three propositions will be helpful later on, all applicable to nonseparable strongly connected directed graphs, of which a source-block in the sign-causality graph is an example.  The first two propositions culminate in the third.

\begin{proposition} \emph{\cite{bondy_graph_2010}} 
\label{prop:FirstEarProposition}
Let $F$ be a (nontrivial)  strongly connected nonseparable proper subgraph of a nonseparable strongly connected directed graph $G$. Then $F$ has a directed ear in $G$. 
\end{proposition}

\begin{proposition}\emph{\cite{bondy_graph_2010}} 
Let F be a strongly connected subgraph of a directed graph G, and let P be directed ear of F in G. Then F $\cup$ P is strongly connected.  
\end{proposition}

	A \emph{directed ear decomposition} of a nonseparable strongly connected directed graph $G$ is a sequence of nonseparable strongly connected directed subgraphs of G, say $G_0, G_1, \dots, G_k$, such that $G_i \subset G_{i+i},\  i = 0,..,k-1$ and such that
\begin{enumerate}
\item $G_0$ is a directed cycle, 
\item $G_{i+1} = G_i \cup P_i$, where $P_i$ is a directed ear of $G_i$ in $G, \;i = 0, \dots, k-1$, 
\item $G_k = G$.
\end{enumerate}

	In effect, the following proposition tells us that a nonseparable strongly connected directed graph --- in particular, a source block --- can be built up by beginning with a directed cycle and successively adding directed ears. (It is left understood that we have in mind graphs with more than one vertex.)

\begin{proposition}\emph{\cite{bondy_graph_2010}} Every nonseparable strongly connected directed graph has a directed ear decomposition.
\end{proposition}

\subsection{Two lemmas about R-to-R and S-to-S intersections of even cycles}

	Here we will argue that when, in the sign-causality graph or the SR Graph, two (not-necessarily directed) even cycles intersect in either an R-to-R path or an S-to-S path, then the third cycle so-formed is also even. We begin by considering the harder case, involving an R-to-R intersection. In the following lemma, when we indicate a path by a symbol such as $R^{*}AR^{**}$, the $A$ does not denote a species. Rather  $R^{*}AR^{**}$ indicates a perhaps long path connecting reactions $R^*$ and $R^{**}$.

\begin{lemma}
\label{lem:R-to-RThreeCycleEven}
Suppose that in the sign-causality graph (or in the SR Graph) there are distinct reactions $R^*$ and $R^{**}$ and three edge-disjoint (not-necessarily-directed) paths, denoted $R^{*}AR^{**}$, $R^{*}BR^{**}$, and $R^{*}CR^{**}$, connecting $R^*$ to $R^{**}$. If the cycles $R^{*}BR^{**}AR^{*}$ and $R^{*}CR^{**}BR^{*}$ are both even, then so is the cycle $R^{*}CR^{**}AR^{*}$.
\end{lemma}

\begin{proof}

	We begin by noting that both $R^*$ and $R^{**}$ will each have adjacent to them three distinct edges. Because there are only two possible labels for each edge (the reactant and product complexes of the adjacent reaction), there must be at least one c-pair centered at $R^*$ and at least one centered at $R^{**}$. (There might be more than one c-pair at $R^*$ or $R^{**}$ if all three adjacent edges carry the same complex label, but this can happen only in chemistries for which three or more species can appear in the same complex; we do not preclude this possibility.)
\begin{figure}[ht]
\centering
\includegraphics[scale=.51]{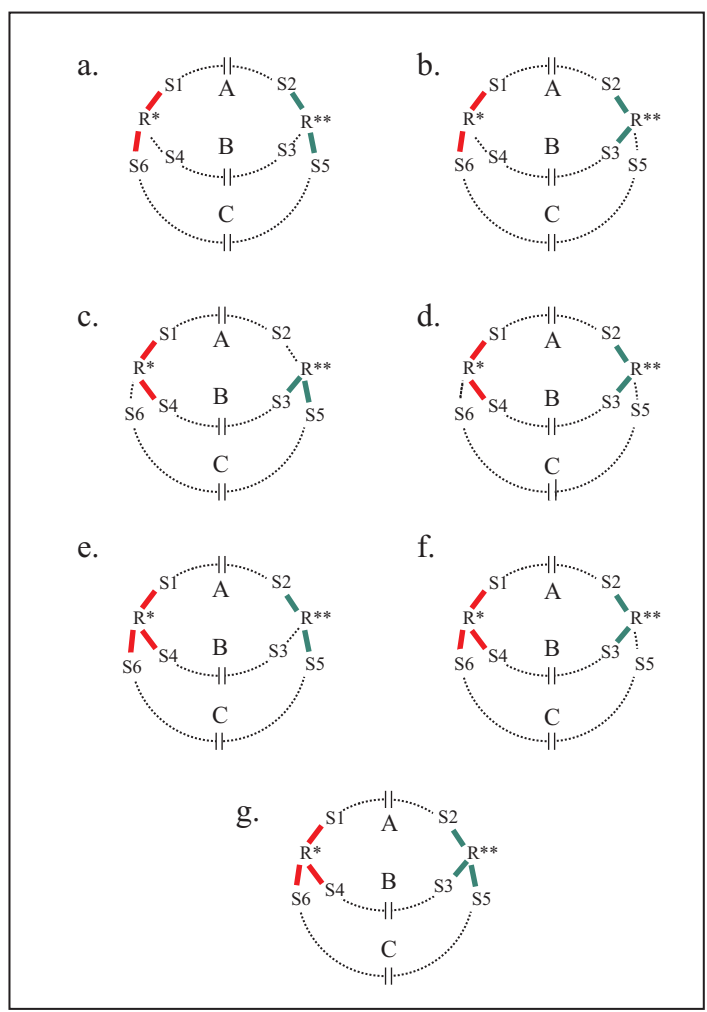}
\caption{Illustration for the proof of Lemma \ref{lem:R-to-RThreeCycleEven}}
\label{fig:R-toRThreeCycleCases}
\end{figure}

	 We show the essentially distinct possibilities for distribution of c-pairs adjacent to $R^*$ and $R^{**}$, relative to the paths $R^{*}AR^{**}$, $R^{*}BR^{**}$, and $R^{*}CR^{**}$, in Figure \ref{fig:R-toRThreeCycleCases}. (We have omitted from the figure cases in which $S_1$ coincides with $S_2$, $S_3$ coincides with $S_4$, or $S_5$ coincides with $S_6$, but analyses of these are not substantially different from those of the cases shown.)    

	The proof essentially amounts to studying each case in Figure \ref{fig:R-toRThreeCycleCases}, beginning with an examination of the relative parity of the number of c-pairs within the species-to-species sub-paths connecting $S_1$ to $S_2$, $S_3$ to $S_4$, and $S_5$ to $S_6$. From these considerations, it can be established in each case that the cycle  $R^{*}CR^{**}AR^{*}$ has an even number of c-pairs.

	Case (b) of the figure is typical: Because the cycle $R^{*}BR^{**}AR^{*}$ is even, because  within the cycle there is a c-pair centered at $R^{**}$, and because within the cycle there is no c-pair  centered at $R^*$, the parity of the number of c-pairs in the $S_1$-to-$S_2$ path must be opposite from the parity of the number of c-pairs within the  $S_3$-to-$S_4$ path.  Because the cycle  $R^{*}CR^{**}BR^{*}$ is even and because there are no c-pairs in the cycle centered at  $R^*$ or  $R^{**}$, the parity of the number of c-pairs in the $S_5$-to-$S_6$ path is identical to the parity of the number of c-pairs in the $S_3$-to-$S_4$ path. Thus the parity of the number of c-pairs in the $S_1$-to-$S_2$ path is opposite to the parity of the number of c-pairs in the $S_5$-to-$S_6
$, and we can conclude that the number of c-pairs in those two disjoint paths is odd. In the cycle  $R^{*}CR^{**}AR^{*}$ there is precisely one additional c-pair --- the one centered at $R^*$ --- so the cycle is even.

	The remaining cases can be examined in the same way.
\end{proof}

\begin{lemma}
\label{lem:S-to-SThreeCycleEven}
Suppose that in the sign-causality graph (or in the SR Graph) there are distinct species $S^*$ and $S^{**}$ and three edge-disjoint (not-necessarily-directed) paths, denoted $S^{*}AS^{**}$, $S^{*}BS^{**}$, and $S^{*}CS^{**}$, connecting $S^*$ to $S^{**}$. If the cycles $S^{*}BS^{**}AS^{*}$ and $S^{*}CS^{**}BS^{*}$ are both even, then so is the cycle $S^{*}CS^{**}AS^{*}$.
\end{lemma}

\begin{proof} Proof of this lemma is similar to the proof of Lemma \ref{lem:R-to-RThreeCycleEven}, but it is substantially simpler: Because c-pairs are centered at reaction vertices, all c-pairs will be interior to the paths  $S^{*}AS^{**}$, $S^{*}BS^{**}$, and $S^{*}CS^{**}$. In fact, the hypothesis ensures that the number of c-pairs along all three paths will be even (respectively, odd) if that number is even (respectively, odd) along any one of them. From this it follows that the cycle $S^{*}CS^{**}AS^{*}$ is even.
\end{proof}

\subsection{Proof of Proposition \ref{prop:SourceBlockSimplicity}} 
\label{subsec:RBlockSBlock}

In this subsection we provide proof of  Proposition \ref{prop:SourceBlockSimplicity}, which is repeated below:
\medskip

\noindent
\textbf{Proposition \ref{prop:SourceBlockSimplicity}.}\; \emph{Suppose that, for the  reaction network under consideration, the Species-Reaction Graph satisfies condition (ii) of Theorem \ref{thm:ConcordanceTheorem}. Then, in any source-block of the sign-causality graph, at most one of the following can obtain:}
\medskip

\noindent
\emph{(i) There is a reaction vertex having more than two adjacent species vertices.}
\smallskip

\noindent
\emph{(ii) There is a species vertex having more than two adjacent reaction vertices.}

\bigskip
	Throughout this subsection we assume that the Species-Reaction Graph of the network under consideration satisfies condition (ii) of Theorem \ref{thm:ConcordanceTheorem}. Because the putative source under consideration, viewed in the SR Graph, is a critical subgraph (Remark \ref{rem:SourceAsCriticalSubgraph}), it follows that no two even cycles in the source have an S-to-R intersection.

	Consider a source-block $G$, which by its very nature is strongly connected and nonseparable, and let $G_0, G_1, \dots, G_k$ be a directed ear decomposition of the block  (\S \ref{sec:Ears}). If the block consists of just a single directed cycle (i.e., if $G = G_0$), then each reaction is adjacent to precisely two species, and each species is adjacent to precisely two reactions. Thus the block is both an S-block and an R-block. 

	Suppose, then, that the block consists of more than a single directed cycle. Then $G_1$ is a directed cycle $G_0$ together with an adjoined directed ear $P_0$ of that cycle residing within the source-block. If that ear joins a reaction vertex to a species vertex of the  cycle, then it is easy to see that $G_1$ will have two directed (and hence even) cycles with an S-to-R intersection, in which case condition (ii) of Theorem \ref{thm:ConcordanceTheorem} would be violated.

	Thus, there are two possibilities: Either the ear $P_0$ is a directed path joining two reactions of the original directed cycle $G_0$ or else  the ear $P_0$ is a directed path that joins two species of that cycle. We shall suppose the former case and then argue in Lemma \ref{lem:R-to-REarDecomposition} that all subsequent ears in the ear-decomposition must also be reaction-to-reaction paths. Thus, at each stage of the decomposition, new edges are adjoined to old vertices only at R-vertices. \emph{In this case, the block can contain no species vertex adjacent to more than two edges. That is, the block is an S-block.}

  \emph{If, on the other hand, the ear $P_0$ is a species-to-species path, then the block is an R-block} by an almost identical argument, via Lemma \ref{lem:S-to-SEarDecomposition} below, with the roles of species and reactions reversed.)

\begin{lemma} 
\label{lem:R-to-REarDecomposition}
Suppose that, in a directed ear decomposition $G_0, G_1, \dots, G_k$ of a source block $G$, no species vertex of $G_i, \; 0 < i < k - 1$,  is adjacent to more than two edges. Then the directed ear $P_i, \; 0 < i < k - 1$, joins a reaction vertex of $G_i$ to another reaction vertex of  $G_i$.
\end{lemma}

\begin{lemma} 
\label{lem:S-to-SEarDecomposition}
Suppose that, in a directed ear decomposition $G_0, G_1, \dots, G_k$ of a source block $G$, no reaction vertex of $G_i, \; 0 < i < k - 1$,  is adjacent to more than two edges. Then the directed ear $P_i, \; 0 < i < k - 1$, joins a species vertex of $G_i$ to another species vertex of  $G_i$.
\end{lemma}

	We provide only the proof of  Lemma \ref{lem:R-to-REarDecomposition}, for the proof of Lemma \ref{lem:S-to-SEarDecomposition} is the same, apart from a reversal of the roles of  species and reactions.

\begin{proof}[Proof of Lemma \ref{lem:R-to-REarDecomposition}] 

	The ear $P_i$ can connect (1) a reaction vertex of $G_i$ with a species vertex of  $G_i$, (2)  a species vertex of $G_i$ with a species vertex of  $G_i$, or (3) a reaction vertex of $G_i$ with a reaction vertex of  $G_i$. We will eliminate the first two possibilities.
\medskip

\noindent
\emph{(Case 1)} Suppose that the directed ear $P_i$ connects a reaction vertex of $G_i$, say $R^*$, with a species vertex of $G_i$, say $S^*$. We will suppose also that  $P_i$ is directed from $S^*$ toward $R^*$; if $P_i$ is oppositely directed the argument in similar.  See Figure \ref{fig:SBlockLemmaArgument}. Because $G_i$ is strongly connected there is a directed path (labeled $A$ in the figure) entirely within $G_i$  connecting $S^*$ to  $R^*$ and also a directed path (labeled $B$) entirely within $G_i$ connecting  $R^*$ to  $S^*$.

\begin{figure}[ht]
\centering
\includegraphics[scale=.48]{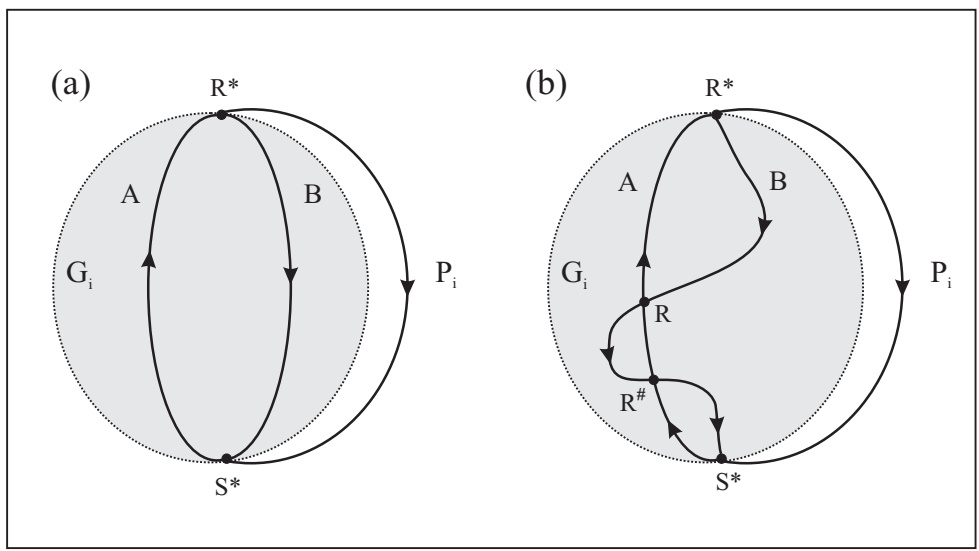}
\caption{Case 1 in the proof of Lemma \ref{lem:R-to-REarDecomposition}}
\label{fig:SBlockLemmaArgument}
\end{figure} 

	If, apart from $S^*$ and $R^*$, these paths have no vertex in common, then they form a directed (and therefore even) cycle within $G_i$, and one of them (path $A$), when taken with the directed ear $P_i$, forms a second directed (and therefore even) cycle that has an S-to-R intersection with the first cycle.  This constitutes a violation of condition (ii) of Theorem \ref{thm:ConcordanceTheorem}.  See Figure \ref{fig:SBlockLemmaArgument}(a).

	Suppose, then, that paths $A$ and $B$ have internal vertices in common. See Figure \ref{fig:SBlockLemmaArgument}(b).\footnote{To simplify the drawing in both Figures \ref{fig:SBlockLemmaArgument} and \ref{fig:SBlockLemmaArgumentPart2} we do not show paths $A$ and $B$ having common edges, but we do not preclude this possibility. The argument for Lemma \ref{lem:R-to-REarDecomposition} does not require that paths $A$ and $B$ be edge-disjoint.} Adjacent to $S^*$ there is an outgoing edge on path $A$ and an incoming edge on path $B$. Thus. the reaction vertices adjacent to $S^*$ on paths $A$ and $B$ are distinct. Now trace path $B$ backward from $S^*$ until its first intersection with path $A$ at a vertex common to both paths. Because at least three edges within $G_i$ meet in this common vertex, it must be a reaction vertex, which we call  $R^{\#}$. We denote by $R^{\#}BS^*$ the directed sub-path along $B$ connecting $R^{\#}$ to  $S^*$ and by  $S^*AR^{\#}$ the directed sub-path along $A$ connecting  $S^*$ to $R^{\#}$. By virtue of their construction, these two sub-paths have no internal vertices in common, so together they form a directed (and therefore even) cycle
 $S^*AR^{\#}BS^*$. Note that this cycle has the S-to-R intersection  $S^*AR^{\#}$ with the directed (and therefore even) cycle formed from path $A$ and the ear $P_i$. This contradicts condition (ii) of Theorem \ref{thm:ConcordanceTheorem}.  Thus, the directed ear $P_i$ cannot connect a reaction vertex of $G_i$ with a species vertex of $G_i$.
\medskip

\noindent
\emph{(Case 2)} Suppose that the directed ear $P_i$ connects a species vertex of $G_i$, say $S^*$, with a different species vertex of $G_i$, say $S^{**}$, and that these vertices have been labeled so that the ear $P_i$ is directed from $S^{*}$ to $S^{**}$.  Because  $G_i$ is strongly connected, there is a directed path (labeled $A$) within $G_i$ that connects  $S^{**}$ to  $S^{*}$ and also a directed path (labeled $B$) within $G_i$ that connects  $S^{*}$ to  $S^{**}$. Should these two paths have internal vertices in common (Figure \ref{fig:SBlockLemmaArgumentPart2}(a)) then, by an argument almost identical to one in Case 1, there would be a contradiction of condition (ii) of Theorem \ref{thm:ConcordanceTheorem}. 

\begin{figure}[ht]
\centering
\includegraphics[scale=.43]{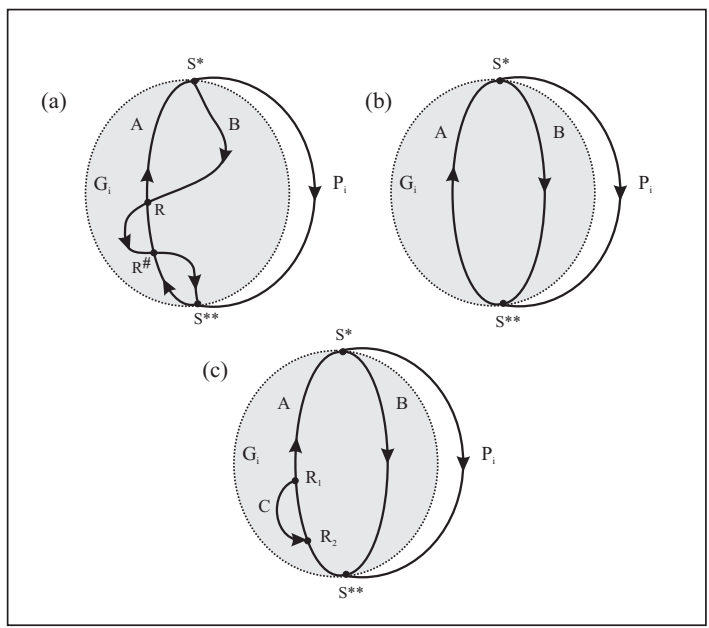}
\caption{Case 2 in the proof of Lemma \ref{lem:R-to-REarDecomposition}}
\label{fig:SBlockLemmaArgumentPart2}
\end{figure} 

	If, however, paths $A$ and $B$ have no internal vertices in common (Figure \ref{fig:SBlockLemmaArgumentPart2}(b)), then this by itself does not constitute a contradiction of condition (ii) of Theorem \ref{thm:ConcordanceTheorem}. On the other hand, $G_i$ is a nonseparable and strongly connected directed graph that is not merely a directed cycle, so the directed cycle $S^{**}AS^{*}BS^{**}$ must have a directed ear within $G_i$, which we label $C$ (Proposition \ref{prop:FirstEarProposition}). Moreover, that ear can only connect two reaction vertices of the cycle because, by hypotheses, within $G_i$ it is only reaction vertices that can be adjacent to more than two edges. (See Figure \ref{fig:SBlockLemmaArgumentPart2}(c).) Thus the cycle $S^{**}AS^{*}BS^{**}$ has two disjoint directed ears, one ($P_i$) connecting a species vertex to a species vertex, and another ($C$) connecting a reaction vertex to a reaction vertex. 

	Note that in Figure \ref{fig:SBlockLemmaArgumentPart2}(c) we show only one of several arrangements whereby the directed R-to-R and S-to-S ears can join to the directed cycle. We show in Figure \ref{fig:SixCases} the essentially different arrangements. (In Figure \ref{fig:SixCases}, $R_O$ and $S_O$ are the cycle-vertices from which the ears are outgoing from the cycle, while $R_I$ and $S_I$ are the cycle-vertices to which the ears are incoming.) In each case there are two cycles --- one containing the R-to-R ear and the other containing the S-to-S ear --- that have an S-to-R intersection. (For readers with access to color these cycles are colored red and green.)  Moreover, each of these cycles is even, either because the cycle is directed or as a consequence of Lemma \ref{lem:R-to-RThreeCycleEven} or \ref{lem:S-to-SThreeCycleEven}. 	Thus, condition (ii) of Theorem \ref{thm:ConcordanceTheorem} is again violated.

	We have shown that Cases 1 and 2 cannot obtain, so $P_i$ must in fact join two R-vertices of $G_i$.
\begin{figure}[tbh]
\centering
\includegraphics[scale=.48]{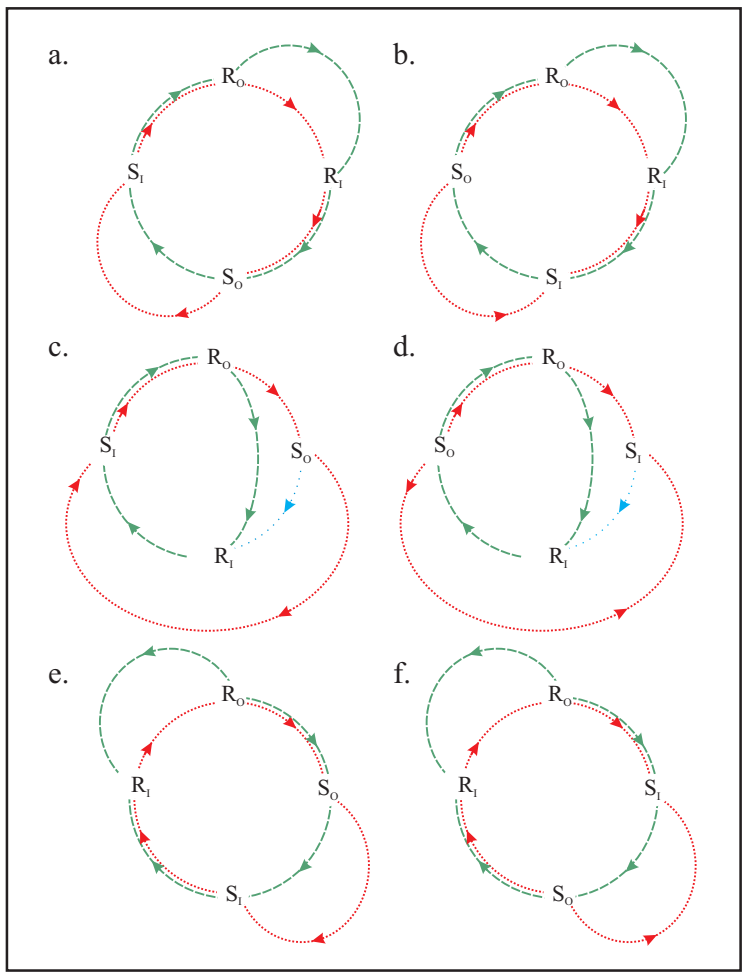}
\caption{A directed cycle with disjoint directed R-to-R and S-to-S ears }
\label{fig:SixCases}
\end{figure}
\end{proof}

	Taken together, Lemmas \ref{lem:R-to-REarDecomposition} and \ref{lem:S-to-SEarDecomposition} ensure that each source-block is either an S-block, an R-block, or a cycle, in which case it is both an S-block and an R-block.

\subsection{Proof of Proposition \ref{prop:EvenCyclesWithinSBlocksRBlocks}}

	Here we prove  Proposition \ref{prop:EvenCyclesWithinSBlocksRBlocks}, which is repeated below: 

\medskip
\noindent
\textbf{Proposition \ref{prop:EvenCyclesWithinSBlocksRBlocks}.} \emph{Every (not necessarily directed) cycle that lies within an S-block or an R-block of a sign-causality graph source is even.} 
\medskip

	We consider a sign-causality graph source-block, $G$, which we will suppose to be an S-block.  (Apart from a reversal of the roles of species and reactions, there is no difference in the proof if the source-block is an R-block.\footnote{For an R-block, however, there is a substantially simpler argument. In an R-block every reaction vertex is adjacent to precisely two species vertices. From this and the fact that the block is strongly connected it follows that every cycle in the block is the union of causal units, so each must be even by an argument given in \S \ref{subsec:EvenCycleSec}.}) If $G$ consists of a single cycle, that cycle is directed and therefore even. Suppose, then, that $G$ consists of more than a single cycle.

	Let $G_0, G_1, \dots, G_k$  be a directed ear decomposition of $G$, so that $G_0$ is a directed (and therefore even) cycle and $G_k = G$. Proceeding inductively, we suppose that for a given value of $i$, $0 \leq i < k$, every cycle in $G_i$ is even. We want to show that $G_{i+1}$ has this same property. Note that  $G_{i+1}$ results from the addition to $G_{i}$ of a directed ear $P_i$, which must join a reaction vertex of $G_{i}$, say $R^{*}$, to another reaction vertex of $G_{i}$, say $R^{**}$. (Recall that $G$ is an S-block, so a species vertex can be adjacent to no more than two edges.)

	Any cycle in $G_{i+1}$ that is entirely within $G_{i}$ is already even by the inductive hypothesis. A cycle in $G_{i+1}$ that is not entirely within $G_{i}$ must contain the ear $P_i$ and is the union of $P_i$ with a not-necessarily-directed path in  $G_{i}$, denoted $A$, that connects $R^{*}$ to $R^{**}$ and that is (necessarily) edge-disjoint from  $P_i$.  Our aim is to show that the cycle $\scrC^{\dagger} := R^{*}P_iR^{**}AR^{*}$ is even. 

	Hereafter, if $C$ is a path in $G$ connecting $R^{*}$ and $R^{**}$, and if $R$ and $R'$ are vertices along that path, we denote by  $RCR'$ the sub-path of $C$ that connects  $R$ with $R'$. No direction is implied.

	Because $G_{i}$ is strongly connected, there is a directed path, labeled $B_0$, that lies entirely within $G_{i}$ and connects $R^{**}$ to $R^{*}$. (See Figure \ref{fig:ProofOfEvenProp}(a).) Note that the cycle $\scrC_0 := R^{*}P_iR^{**}B_0R^{*}$ is even because it is directed. If the path $B_0$ has no edge in common with the path $A$, then  the cycle  $R^{*}AR^{**}B_0R^{*}$, which lies entirely within $G_{i}$ is also even by the inductive hypothesis. In this case, Lemma \ref{lem:R-to-RThreeCycleEven} ensures that the cycle $\scrC^{\dagger} = R^{*}P_iR^{**}AR^{*}$ is even, which is what we wanted to show.

	On the other hand, suppose that the path $B_0$ does have one or more edges in common with the path $A$. If  $B_0$ coincides with $A$ then the cycle $\scrC^{\dagger} = R^{*}P_iR^{**}AR^{*}$ coincides with the even cycle $\scrC_0= R^{*}P_iR^{**}B_0R^{*}$, in which case we are finished.

\begin{figure}[tbh]
\centering
\includegraphics[scale=.49]{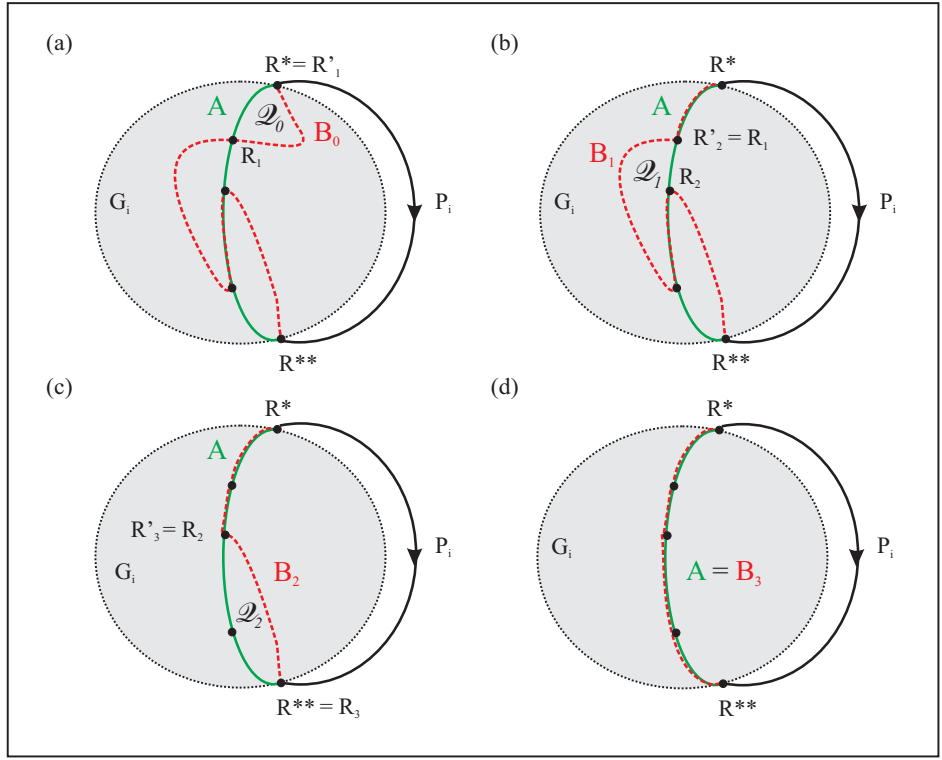}
\caption{Schematic illustration for proof of Proposition \ref{prop:EvenCyclesWithinSBlocksRBlocks} }
\label{fig:ProofOfEvenProp}
\end{figure}

	Suppose, therefore, that the path $A$ contains an edge that is not contained in path $B_0$, so the cycle $\scrC^{\dagger}$ is different from the even cycle $\scrC_0$. In this case we will argue that there is a sequence of cycles $\scrC_0, \scrC_1, \scrC_2,\dots,\scrC_p$ with  the following properties:

\begin{enumerate}[(i)]
\item $\scrC_{\theta}$ is even, $\theta = 0,\dots, p$.

\item $\scrC_{\theta}$ is the union of $P_i$ and a path $B_{\theta}$ that connects $R^{**}$ with $R^{*}$ and that lies entirely within $G_{i}$. 

\item The path $B_{\theta}$ contains a sub-path $R^*{B}_{\theta}R_{\theta}$ that lies entirely within path $A$; that is, $R^*{B}_{\theta}R_{\theta} = R^*AR_{\theta}$.\footnote{For $\theta = 0$, the sub-path mentioned might be the trivial one containing only the vertex $R^*$, in which case $R_0 = R^*$} Moreover, for $\theta > 0$, $R^*{B}_{\theta - 1}R_{\theta - 1}$ (= $R^*AR_{\theta - 1}$) is properly contained within $R^*{B}_{\theta}R_{\theta}$ (= $R^*AR_{\theta}$).

\item $B_p = A$ (so that $\scrC_{p} = \scrC^{\dagger}$, whereupon $\scrC^{\dagger}$ is even).  
\end{enumerate}

	Because of the construction of $\scrC_0$ it satisfies  conditions (i)-(iii).   Suppose, then, that for a particular $\theta > 0$, $\scrC_{\theta}$ satisfies conditions (i)-(iii). We want to show that if path $B_{\theta}$ is not identical to path $A$, then there is a cycle $\scrC_{\theta + 1}$ (and, in particular, a path $B_{\theta + 1}$) such that the requirements of (i)-(iii) are met. 

	Suppose that path $B_{\theta}$ is not identical to path $A$. Moving away from $R^*$ along path $A$, choose the first edge in $A$ that is not in $B_{\theta}$, and consider the longest sub-path of $A$ that contains this edge and that has no  edge in common with $B_{\theta}$. It is not difficult to see that the end-vertices of this sub-path must  be members of both $A$ and $B_{\theta}$ and that each such end-vertex will be adjacent to at least three edges of $G_{i+1}$. Because $G$ is an S-block, those end-vertices must be reaction vertices, denoted $R_{\theta + 1}$ and $R'_{\theta + 1}$, with $R_{\theta + 1}$ denoting the end-vertex furthest from $R^{*}$ along path $A$. (See Figure \ref{fig:ProofOfEvenProp}(a)-(c).) Note that the sub-path $R^*B_{\theta}R_{\theta}$ (= $R^*AR_{\theta}$) must be contained properly within
$R^*AR_{\theta + 1}$. Were it not, then $R'_{\theta + 1}AR_{\theta + 1}$ would lie entirely in  $B_{\theta}$, contradicting what has been assumed.

	Because no edge of $R_{\theta + 1}{A}R'_{\theta + 1}$ is an edge of $B_{\theta}$, the sub-paths $R_{\theta + 1}{B}_{\theta}R'_{\theta + 1}$ and $R_{\theta + 1}{A}R'_{\theta + 1}$ are edge-disjoint. Together they form a cycle,  \[\mathscr{Q_{\theta}} := R_{\theta + 1}{B}_{\theta}R'_{\theta + 1}A R_{\theta + 1},\] that is contained entirely within $G_{i}$ and is therefore even by the inductive hypothesis. 

	Note that the even cycle $\scrC_{\theta}$ is not only the union of $P_i$ and $B_{\theta}$, it is also he union of the path  $R'_{\theta + 1}B_{\theta}R_{\theta + 1}$ and its complementary path within $\scrC_{\theta}$, the one connecting $R_{\theta + 1}$ to $R'_{\theta + 1}$ via path $P_i$.  This last path we denote $R_{\theta + 1}DR'_{\theta + 1}$. Thus, we can write $\scrC_{\theta} = R'_{\theta + 1}B_{\theta}R_{\theta + 1}DR'_{\theta + 1}$. Apart from vertices and edges that lie in $P_i$, all vertices and edges in $R_{\theta + 1}DR'_{\theta + 1}$ reside in $B{_\theta}$.

	Because the paths $R_{\theta + 1}{A}R'_{\theta + 1}$,  $R_{\theta + 1}{B}_{\theta}R'_{\theta + 1}$ and $R_{\theta + 1}DR'_{\theta + 1}$ are edge-disjoint and because the cycles $\scrC_{\theta}$ and $\mathscr{Q}_{\theta}$ are even, it follows from Lemma \ref{lem:R-to-RThreeCycleEven} that the cycle $\scrC_{\theta + 1} := R'_{\theta + 1}AR_{\theta + 1}DR'_{\theta + 1}$ is even. Thus, the cycle $\scrC_{\theta + 1}$ satisfies condition (i) above.

	Note that cycle $\scrC_{\theta + 1}$  is union of $P_i$ with the path 
\begin{eqnarray*}
B_{\theta + 1} &:=& R^{*}B_{\theta}R'_{\theta +1}AR_{\theta +1}B_{\theta}R^{**}\\
&=& R^{*}AR'_{\theta +1}AR_{\theta +1}B_{\theta}R^{**}\\
&=& R^{*}AR_{\theta +1}B_{\theta}R^{**}.
\end{eqnarray*}
Thus, $\scrC_{\theta + 1}$ satisfies condition (ii). Note too that, as indicated earlier,  $R^{*}B_{\theta}R_{\theta}$ (= $R^{*}AR_{\theta}$) is properly contained within  $R^{*}AR_{\theta +1}$. Thus, condition (iii) is also satisfied.

	So long as $B_{\theta}$ is not identical to $A$ these iterations can continue, with every new $B_{\theta + 1}$ obtained by replacing a sub-path $R^*B_{\theta}R_{\theta + 1}$ of $B_{\theta}$ with a sub-path $R^{*}AR_{\theta + 1}$ of $A$, with each such sub-path of $A$ strictly longer than its predecessor,  $R^*AR_{\theta}$. Because $A$ has a finite number of edges, the process must terminate with some $B_p$ identical to $A$.

	Thus, the cycle $\scrC^{\dagger}$, which is identical to   $\scrC_p$, is even.

\section{Proof of Proposition \ref{prop:RBlockMExistence}.}
\label{app:RBlockAppendix}

	In this appendix we prove Proposition \ref{prop:RBlockMExistence}, which is repeated below. The argument here has interesting connections to the mathematics in \cite{feinberg_necessary_1989}, a paper about detailed balance in mass action systems.
\medskip

\noindent
\textbf{Proposition \ref{prop:RBlockMExistence}}. \emph{Suppose that, in a sign-causality graph source, an \mbox{R-block} has species set $\scrS^*$, and suppose that no directed cycle in the block is stoichiometrically expansive relative to the $\rightsquigarrow$ orientation. Then there is a set of positive numbers $\{M_s\}_{s \; \in \scrS^*}$ such that, for each causal unit $s \rightsquigarrow R \rightsquigarrow s'$ in the block, the following relation is satisfied:}
\begin{equation}
f_{R \rightsquigarrow s'}M_{s'} - e_{s \rightsquigarrow R}M_{s} \leq 0.
\end{equation}

\begin{proof} Hereafter we let $\mathscr{U}$ denote the set of causal units in the R-block under consideration. Given the hypothesis, we want to show the existence of positive numbers $\{M_s\}_{s \; \in \scrS^*}$ such that

\begin{equation}
\frac{M_{s'}}{M_{s}} \;\leq \; \frac{e_{s \leadsto R}}{f_{R \leadsto s'}}, \quad \forall \; s  \leadsto R \leadsto s' \; \in \mathscr{U}
\end{equation}
or, equivalently,
\begin{equation}
\label{eq:MCondition}
\ln M_{s'} - \ln M_{s} \; \leq \; \ln\, \frac{e_{s \leadsto R}}{f_{R \leadsto s'}}, \quad \forall \; s  \leadsto R \leadsto s' \; \in \;\mathscr{U}.
\end{equation}

	Now let $\{\omega_s \}_{s \in \scrSst}$ be the standard basis for \RSst, let $\{\epsilon_{s \leadsto R \leadsto s'}\}_{s \leadsto R \leadsto s' \;\in\; \scrU}$ be the standard basis for \RU, let $T: \RSst \to \RU$ be defined by 
\begin{equation}
Tq := \sum_{ s  \leadsto R \leadsto s' \; \in \;\mathscr{U}}(\omega_{s'} - \omega_{s})\cdot q \;\epsilon_{s \leadsto R \leadsto s'},
\end{equation}
and let $z \in \RU$ be defined by
\begin{equation}
z := \sum_{ s  \leadsto R \leadsto s' \; \in \;\mathscr{U}}\ln\, \frac{e_{s \leadsto R}}{f_{R \leadsto s'}}\epsilon_{s \leadsto R \leadsto s'}.
\end{equation}
The existence of  positive numbers $\{M_s\}_{s \; \in \scrS^*}$  satisfying \eqref{eq:MCondition} is readily seen to be equivalent to the existence of $q^* \in \RSst$ satisfying
\begin{equation}
Tq^* \leq z,
\end{equation}
for then we can take $M_s = \exp q^*_s, \; \forall s \in \scrSst$ to meet the requirements of \eqref{eq:MCondition}.

	However, by a theorem of Gale \cite[p. 46] {Gale1960}, the existence of such a $q^*$ is equivalent to the \emph{non-existence} of $p \;\in\; \ker T^T \cap \;\PbarU$ satisfying
\begin{equation}
\label{eq:GaleConditionOnP}
p \cdot z < 0.
\end{equation} 
In our case, $T^T: \RU \to \RSst$ is given by
\begin{equation}
T^Tx := \sum_{ s  \leadsto R \leadsto s' \; \in \;\mathscr{U}}x_{ s  \leadsto R \leadsto s'}(\omega_{s'} - \omega_{s}).
\end{equation}

	We will say that $c \in \PbarU$ is a \emph{directed cycle vector} if, in the R-block under study, there is a directed cycle such that $c_{s' \leadsto R \leadsto s} = 1$ if $s' \leadsto R \leadsto s$ is a causal unit in the cycle and is zero otherwise. It is easy to see that every directed cycle vector is a member of $\ker T^T \cap \;\PbarU$. In fact, as we show in Lemma \ref{lem:DirectedCycleLemma} below,  every nonzero member of $\ker T^T \cap \;\PbarU$ is a positive linear combination of directed cycle vectors. 

	Suppose that $p \;\in\; \ker T^T \cap \;\PbarU$ satisfies condition \eqref{eq:GaleConditionOnP}, and let
\begin{equation}
p = \sum_{\theta = 1}^k\alpha_{\theta}c_{\theta}
\end{equation}
be a representation of $p$ as a positive linear combination of directed cycle vectors. Then we must have 
\begin{equation}
\label{eq:GaleConditionExpanded}
p \cdot z = \sum_{\theta = 1}^k\alpha_{\theta}(c_{\theta}\cdot z) < 0, 
\end{equation}
where all the $\alpha_{\theta}$ are positive. On the other hand, \eqref{eq:GaleConditionExpanded} is contradicted by the fact that each $c_{\theta}\cdot z$ is non-negative: To see this, suppose that, for a particular $\theta$, the directed cycle corresponding to $c_{\theta}$ is 
\begin{equation}
\label{eq:SourceDirectedCycle2}
s_1 \rightsquigarrow R_1  \rightsquigarrow s_2 \rightsquigarrow R_2 \dots\rightsquigarrow s_{n} \rightsquigarrow R_n \rightsquigarrow s_1.
\end{equation}
Then 
\begin{equation}
c_{\theta}\cdot z = \sum_{i = 1}^n \ln \frac{e_{s_i \leadsto R_i}}{f_{R_i \leadsto s_{i+1}}} = - \ln \prod_{i = 1}^n   \frac{f_{R_i \leadsto s_{i+1}}}{e_{s_i \leadsto R_i}} \geq 0,
\end{equation}
with the last inequality coming from the fact that no directed cycle in the R-block is stoichiometrically expansive. Apart from proof of Lemma \ref{lem:DirectedCycleLemma}, which appears below, this concludes the proof of Proposition \ref{prop:RBlockMExistence}.
\end{proof}

\begin{lemma}
\label{lem:DirectedCycleLemma}
Every nonzero member of $\ker T^T \cap \;\PbarU$ is a positive linear combination of directed cycle vectors.
\end{lemma}
\begin{proof} Here we use the terminology and results in \cite[Chapter 2]{Gale1960}. A vector in $\ker T^T \cap \;\PbarU$ is \emph{extreme} if it cannot be written as the sum of two linearly independent vectors of  $\ker T^T \cap \;\PbarU$. Because $\ker T^T \cap \;\PbarU$ is a pointed finite cone (sometimes called a pointed polyhedral cone), any vector within it is the sum of extreme vectors of  $\ker T^T \cap \;\PbarU$. We will be finished if we can show that every nonzero extreme vector of  $\ker T^T \cap \;\PbarU$ is a positive multiple of a directed cycle vector.

	Suppose, then, that $x^*$ is a nonzero extreme vector of $\ker T^T \cap \;\PbarU$.  Thus, $x^*$ must be non-negative and satisfy the equation
\begin{equation}
T^Tx^* = \sum_{ s  \leadsto R_{\theta} \leadsto s' \; \in \;\mathscr{U}}x^*_{ s  \leadsto R \leadsto s'}(\omega_{s'} - \omega_{s}) = 0.
\end{equation}
Let $\scrU^*$ denote the support of $x^*$, and let $\RU_*$ be the linear subspace of \RU consisting of all vectors whose support is a subset of  $\scrU^*$. That is,
\begin{equation}
\RU_* := \{x \in \RU: x_{s \leadsto R \leadsto s'} =0 \; \mathrm{when} \; {s \leadsto R \leadsto s'} \notin \scrU^*\}.
\end{equation}
Moreover, let \PbarUst\; denote the set of vectors of $\RU_*$ having exclusively non-negative components. We denote by $T_*^T: \RU_* \to \RS$ the linear map defined by
\begin{equation}
T_*^Tx := \sum_{ s  \leadsto R \leadsto s' \; \in \;\mathscr{U}^*}x_{ s  \leadsto R \leadsto s'}(\omega_{s'} - \omega_{s}).
\end{equation} 
Because the support of $x^*$ is $\scrU^*$, it is easy to see that $T_*^Tx^* = T^Tx^* = 0$. Thus, the dimension of $\ker T_*^T$ is at least one. 

	In fact if, as has been supposed, $x^*$ is an extreme vector of  $\ker T^T \cap \;\PbarU$, then the dimension of $\ker T_*^T$ must be precisely one:
If not, then there is nonzero vector $d \in  \ker T_*^T \subset \RU_*$ that is not collinear with $x^*$.  By choosing $\epsilon > 0$ sufficiently small, we can construct another vector $x^1 := x^* + \epsilon d$ that also resides in  $\ker T_*^T$ and that has $x^1_{s \leadsto R \leadsto s'} > 0$ for each $s  \leadsto R \leadsto s' \in \scrU^*$. Because $d$ is not collinear with $x^*$, neither is $x^1$. By choosing $\rho > 0$ sufficiently small, the vector $x^2 := x^* - \rho x^1 \in \ker T_*^T$ can also be made to have positive components on $\scrU^*$. Since both $x_1$ and $x_2$ are members of \PbarUst (and hence of \PbarU) and since both are members of $\ker T_*^T$ (and hence of $\ker T^T$), both are members of  $\ker T^T \cap \;\PbarU$. Moreover, since $x^* = \rho x_1 + x_2$, $x^*$ is the sum of two vectors in $\ker T^T \cap \;\PbarU$ that are not collinear with it (and, therefore, not collinear with each other). This implies that $x^*$ is not extreme, which contradicts what we had supposed. Thus, $\dim \ker T_*^T = 1$.

	Now $ T_*^T$ can be identified with the \emph{incidence map} \cite{Biggs1993} of a directed graph $\bar{G}$ having as its vertices the species of $\scrU^*$ and as edges the directed causal units contained in $\scrU^*$. It follows from standard graph theoretical arguments --- see Chapters 4 and 5 in \cite{Biggs1993} --- that the kernel of the incidence map of a directed graph (when it is singular) has a basis consisting of \emph{cycle vectors}. Each vector $c$ of such a basis has support on the edges of a distinct (not necessarily directed) cycle of the graph and has components related to the cycle in following way:  The cycle is given an orientation (e.g. ``clockwise") and the component of $c$ corresponding to a particular edge in the cycle is $+1$ [respectively, $-1$] if the edge's direction agrees [disagrees] with the chosen orientation.

	Because  $\dim \ker T_*^T = 1$, every vector in  $\ker T_*^T$ is a multiple of a single cycle vector $c^*$, corresponding to a cycle in $\bar{G}$ made up of edges (causal units) in $\scrU^*$. Because $x^*$ is a member  $\ker T_*^T$, it is a non-zero multiple of $c^*$. Moreover, because the components of  $x^*$ corresponding to members of $\scrU^*$ are all positive, it must be the case that all components of  $c^*$ corresponding to members of  $\scrU^*$ also have the same sign, which is to say that the directions of members of  $\scrU^*$ are consistent with a fixed cycle orientation (clockwise or counterclockwise). That is, the cycle  in $\bar{G}$ formed by members of  $\scrU^*$ is directed. Thus, we can write $x^* = \alpha c^*$, where $\alpha$ is positive and $c^*$ is the directed cycle vector corresponding to the underlying cycle in $\bar{G}$, taken with its native $\leadsto$ orientation.
\end{proof}

\section{Proof of Lemma \ref{lem:StrongConcordanceLemma}}
\label{app:ProofofStrongConcordanceLemme}

	Here we prove Lemma \ref{lem:StrongConcordanceLemma}, which is repeated below:

\medskip
\noindent
\textbf{Lemma \ref{lem:StrongConcordanceLemma}}. \emph{Suppose that a fully open reaction network with true chemistry \newline  \rnet\  is not strongly concordant. Then there is another true chemistry \rnetbar\ whose fully open extension is discordant and whose SR Graph is identical to a subgraph of the SR Graph for \rnet, apart perhaps from changes in certain arrow directions within the reaction vertices.}
\medskip

	Suppose that, for a fully open network, the true chemistry \rnet\ is not strongly concordant. We let $\scrR_{aug}$ denote the set of all reactions in the fully open network whose true chemistry is \rnet, including synthesis-degradation reactions of the form $0 \to s$ and $s \to 0$. To say that the fully open network is not strongly concordant is to say that there is a non-zero $\sigma \in \RS$ and numbers $\{\alpha_{\rxn}\}_{\rxn \in \scrR_{aug}}$ such that

\begin{equation}
\label{eq:AlphEqnForUnbarNetwork}
\sum_{\rxn \in \scrR_{aug}}\alpha_{\rxn}(y' - y) = 0,
\end{equation}
\noindent
and, in addition,

\begin{enumerate}[(i)]
\item{For each $y \to y' \in \scrR_{aug}$ such that $\alpha_{y \to y'} > 0$ there exists a species $s \in \supp (y - y')$ for which $\sgn \sigma_s = \sgn (y - y')_s$.}
\item{For each $y \to y' \in \scrR_{aug}$ such that $\alpha_{y \to y'} < 0$ there exists a species $s \in \supp (y - y')$ for which $\sgn \sigma_s = - \sgn (y - y')_s$.}

\item{For each $y \to y' \in \scrR_{aug}$ such that $\alpha_{y \to y'} = 0$, either (a) $\sigma_s = 0$ for all $s \in \supp y$, or (b) there exist species $s, s' \in \supp (y - y')$ for which $\sgn \sigma_s = \sgn (y - y')_s $ and $\sgn \sigma_{s'} = -  \sgn (y - y')_{s'}$.}
\end{enumerate}

	Our aim will be to construct a true chemistry \rnetbar\ with a discordant fully open extension and with  $\bar{\scrR}$ resulting from removal of certain reactions in \scrR, retention of certain other reactions in \scrR,  and replacement of still other reactions in \scrR\, by their reverses. It is easy to see that the resulting true chemistry \rnetbar will have an SR Graph that is identical to a subgraph of the SR Graph for \rnet, up to modification of arrow directions within the reaction vertices. 

To produce the discordance we will want to show that, with $\bar{\scrR}_{aug}$ denoting the reaction set $\bar{\scrR}$ augmented by the degradation reactions $\{s \to 0, \forall s \in \scrS\}$, there exist numbers  $\{\bar{\alpha}_{\rxn}\}_{\rxn \in \bar{\scrR}_{aug}}$ satisfying

\begin{equation}
\label{eq:AlphaEqnForConstuctedBarNetwork}
\sum_{\rxn \in \bar{\scrR}_{aug}}\bar{\alpha}_{\rxn}(y' - y) = 0,
\end{equation}

\medskip
\noindent
and, in addition (with $\sigma$ as before),
\medskip
\begin{enumerate}[(i)]
\item{For each $\rxn \in \bar{\scrR}_{aug}$ such that $\bar{\alpha}_{\rxn} \neq 0$, \supp $y$ contains a species $s$ for which $\mathrm{sgn}\, \sigma_s = \mathrm{sgn}\, \bar{\alpha}_{\rxn}$.}

\item{For each  $\rxn \in \bar{\scrR}_{aug}$ such that $\bar{\alpha}_{\rxn} = 0$, $\sigma_s = 0$\, for all $s \in \supp y$ or else \supp $y$ contains species $s\, \textrm{and } s'$ for which  $\mathrm{sgn}\, \sigma_s = -\, \mathrm{sgn}\; \sigma_{s'}$, both not zero.}
\end{enumerate}

\noindent
The resulting fully open reaction network, with reaction set $\bar{\scrR}_{aug}$, will then be discordant. (Recall Definition \ref{def:Concordant}, keeping in mind that the stoichiometric subspace of the fully open network is \RS.)

	The construction of $\bar{\scrR}_{aug}$, which draws on conditions (i) -- (iii) just below \eqref{eq:AlphEqnForUnbarNetwork}, proceeds through sequential retention, removal, or reversal of reactions in $\scrR_{aug}$ in a procedure we will now describe, along with an indication of how the numbers $\{\bar{\alpha}_{\rxn}\}_{\rxn \in \bar{\scrR}_{aug}}$ are to be assigned:

\medskip
\noindent
1. For each $\rxn \in \scrR$ that is irreversible do the following:

\begin{enumerate}[(A)]
\item{If $\alpha_{\rxn} \neq 0$ and there exists $s \in \supp y$ with $\sgn \sigma_s = \sgn \alpha_{\rxn}$ retain \rxn\, and set $\bar{\alpha}_{\rxn} = \alpha_{\rxn}$.} 

\item{If $\alpha_{\rxn} \neq 0$ and there does not exist $s \in \supp y$ with $\sgn \sigma_s = \sgn \alpha_{\rxn}$ (in which case there exists  $s' \in \supp y'$ with $\sgn \sigma_{s'} = - \sgn \alpha_{\rxn}$)
replace \rxn\ by $y' \to y$, and set $\bar{\alpha}_{y' \to y} = -\alpha_{\rxn}$.} 

\item{If $\alpha_{\rxn} = 0$ remove \rxn.}
\end{enumerate}

\medskip
\noindent
2. For each reversible reaction pair $y \rightleftarrows y' \in \scrR$ do the following:

\begin{enumerate}[(A)]
\item{If $\alpha_{\rxn} \neq \alpha_{y' \to y}$ and the complexes have been labeled so that $|\alpha_{\rxn}| > |\alpha_{y' \to y}|$ (whereupon $\sgn (\alpha_{\rxn} - \alpha_{y' \to y}) = \sgn \alpha_{\rxn}$) } then
\begin{enumerate}[(i)]
\item{If there exists $s \in \supp y$ with $\sgn \sigma_s = \sgn \alpha_{\rxn}$ then retain \rxn, remove $y' \to y$, and set $\bar{\alpha}_{\rxn} = \alpha_{\rxn} - \alpha_{y' \to y}$.} 

\item{If there does not exist $s \in \supp y$ with $\sgn \sigma_s = \sgn \alpha_{\rxn}$ (in which case there exists  $s' \in \supp y'$ with $\sgn \sigma_{s'} = - \sgn \alpha_{\rxn}$) then retain $y' \to y$, remove \rxn,
and set $\bar{\alpha}_{y' \to y} = \alpha_{y' \to y} - \alpha_{\rxn}$.} 
\end{enumerate}
\item{If $\alpha_{\rxn} = \alpha_{y' \to y}$ then remove both \rxn\, and $y' \to y$.}
\end{enumerate}

\medskip
\noindent
3. For each degradation reaction $s \to 0$ that is irreversible in $\scrR_{aug}$, retain $s \to 0$ and set $\bar{\alpha}_{s\to 0} = \alpha_{s\to 0}$.

\medskip
\noindent
4. For each reversible pair $s \rightleftarrows 0 \in \scrR_{aug}$, do the following:
\begin{enumerate}[(A)]
\item{If $\alpha_{s \to 0} \neq 0$ (in which case $\alpha_{0 \to s} = 0$ or $\sgn \alpha_{0 \to s} = -\sgn \alpha_{s \to 0}$), retain $s \to 0$, remove $0 \to s$, and set $\bar{\alpha}_{s \to 0} = \alpha_{s \to 0} - \alpha_{0 \to s}$}.
\item{If $\alpha_{s \to 0} = 0$ (in which case $\alpha_{0 \to s} = 0$), retain $s \to 0$, remove $0 \to s$, and set $\bar{\alpha}_{s \to 0} = 0$.}  
 \end{enumerate}
\bigskip

	That \eqref{eq:AlphaEqnForConstuctedBarNetwork} is satisfied follows from the construction and \eqref{eq:AlphEqnForUnbarNetwork}. That the two conditions just below \eqref{eq:AlphaEqnForConstuctedBarNetwork} are satisfied also follows from the construction.

\section*{Acknowledgment} We are grateful to Daniel Knight for helpful discussions and to Uri Alon for his support and encouragement of this work.

\medskip
\noindent

\bibliographystyle{spmpsci}
\bibliography{MonoLibrary}

\begin{thebibliography}{10}
\providecommand{\url}[1]{{#1}}
\providecommand{\urlprefix}{URL }
\expandafter\ifx\csname urlstyle\endcsname\relax
  \providecommand{\doi}[1]{DOI~\discretionary{}{}{}#1}\else
  \providecommand{\doi}{DOI~\discretionary{}{}{}\begingroup
  \urlstyle{rm}\Url}\fi

\bibitem{banaji_graph-theoretic_2009}
Banaji, M., Craciun, G.: Graph-theoretic approaches to injectivity and multiple
  equilibria in systems of interacting elements.
\newblock Communications in Mathematical Sciences \textbf{7}(4), 867--900
  (2009)

\bibitem{banaji_graph-theoretic_2010}
Banaji, M., Craciun, G.: Graph-theoretic criteria for injectivity and unique
  equilibria in general chemical reaction systems.
\newblock Advances in Applied Mathematics \textbf{44}(2), 168--184 (2010)

\bibitem{Biggs1993}
Biggs, N.: Algebraic Graph Theory, second edn.
\newblock Cambridge University Press (1993)

\bibitem{bondy_graph_2010}
Bondy, A., Murty, U.: Graph Theory.
\newblock Springer (2010)

\bibitem{craciun_multiple_2005-1}
Craciun, G., Feinberg, M.: Multiple equilibria in complex chemical reaction
  networks. {I.} the injectivity property.
\newblock {SIAM} Journal on Applied Mathematics \textbf{65}, 1526--1546 (2005)

\bibitem{craciun_multiple_2006-2}
Craciun, G., Feinberg, M.: Multiple equilibria in complex chemical reaction
  networks: extensions to entrapped species models.
\newblock {IEE} Proc. Syst. Biol \textbf{153}, 179--186 (2006)

\bibitem{craciun_multiple_2006-1}
Craciun, G., Feinberg, M.: Multiple equilibria in complex chemical reaction
  networks. {II.} the species-reaction graph.
\newblock {SIAM} Journal on Applied Mathematics \textbf{66}, 1321--1338 (2006)

\bibitem{craciun_multiple_2010}
Craciun, G., Feinberg, M.: Multiple equilibria in complex chemical reaction
  networks: Semiopen mass action systems.
\newblock {SIAM} Journal on Applied Mathematics \textbf{70}(6), 1859--1877
  (2010)

\bibitem{craciun_understanding_2006-1}
Craciun, G., Tang, Y., Feinberg, M.: Understanding bistability in complex
  enzyme-driven reaction networks.
\newblock Proceedings of the National Academy of Sciences \textbf{103},
  8697--8702 (2006)

\bibitem{feinberg_lectureschemical_1979}
Feinberg, M.: Lectures on chemical reaction networks (1979).
\newblock Written version of lectures given at the Mathematical Research
  Center, University of Wisconsin, Madison, {WI} Available at
  \url{www.chbmeng.ohio-state.edu/~feinberg/LecturesOnReactionNetworks}

\bibitem{feinberg_necessary_1989}
Feinberg, M.: Necessary and sufficient conditions for detailed balancing in
  mass action systems of arbitrary complexity.
\newblock Chemical Engineering Science \textbf{44}, 1819--1827 (1989)

\bibitem{feinberg_chemical_1977}
Feinberg, M., Horn, F.J.M.: Chemical mechanism structure and the coincidence of
  the stoichiometric and kinetic subspaces.
\newblock Archive for Rational Mechanics and Analysis \textbf{66}, 83--97
  (1977)

\bibitem{Gale1960}
Gale, D.: The Theory of Linear Economic Models.
\newblock University of Chicago Press (1960)

\bibitem{Ji_Thesis}
Ji, H.: Uniqueness of equilibria for complex chemical reaction networks.
\newblock Ph.D. thesis, Department of Mathematics, The Ohio State University
  (2011)

\bibitem{Ji_toolboxWinV21}
Ji, H., Ellison, P., Knight, D., Feinberg, M.: The chemical reaction network
  toolbox, version 2.1 (2011).
\newblock Available at
  \url{http://www.chbmeng.ohio-state.edu/~feinberg/crntwin/}

\bibitem{Schlosser_PhDThesis}
Schlosser, P.M.: A graphical determination of the possibility of multiple
  steady states in complex isothermal {CFSTRs}.
\newblock Ph.D. thesis, University of Rochester (1988)

\bibitem{schlosser_theory_1994}
Schlosser, P.M., Feinberg, M.: A theory of multiple steady states in isothermal
  homogeneous {CFSTRs} with many reactions.
\newblock Chemical Engineering Science \textbf{49}(11), 1749--1767 (1994)

\bibitem{shinar_concordant_2011}
Shinar, G., Feinberg, M.: Concordant chemical reaction networks.
\newblock {arXiv:1109.2923}  (2011).
\newblock \urlprefix\url{http://arxiv.org/abs/1109.2923}.
\newblock Under revision for Mathematical Biosciences at the request of the
  editor (2012).

\end{thebibliography}

\end{document}